\documentclass[11pt,a4paper]{article}
\usepackage[margin=1in]{geometry}
\usepackage{enumitem}
\usepackage{amsmath}
\usepackage{graphicx} 
\usepackage{amsthm}
\usepackage{amssymb}
\usepackage{algorithm}
\usepackage{algorithmic}
\usepackage{color}
\usepackage{tikz}
\usepackage{subcaption}
\usepackage{multirow}
\usepackage{natbib}
\usepackage{url}
\usepackage{todonotes}

\usepackage{soul}
\usepackage{orcidlink}
\usepackage{circledsteps}

\pgfkeys{/csteps/inner ysep=10pt}
\pgfkeys{/csteps/inner xsep=10pt}
\pgfkeys{/csteps/inner color=red}
\pgfkeys{/csteps/outer color=red}

\title{Course Allocation with Credits via Stable Matching\thanks{A preliminary version of this paper appeared in the Proceedings of SAGT 2025: the 18th International Symposium on Algorithmic Game Theory \cite{rodriguez_course_2025}.}}
\author{José Rodríguez \orcidlink{0009-0008-2470-6549}, David Manlove \orcidlink{0000-0001-6754-7308} \vspace{4mm}
\\
\small
\emph{School of Computing Science, University of Glasgow, Glasgow, UK.}
\\
\small
\emph{Email {\tt j.rodriguez-bacallado.1@research.gla.ac.uk}, {\tt david.manlove@glasgow.ac.uk}.}
}
\date{ }

\begin{document}

\newtheorem{theorem}{Theorem}
\newtheorem{proposition}{Proposition}
\newtheorem{lemma}{Lemma}
\newtheorem{claim}{Claim}
\newtheorem{observation}{Observation}
\newtheorem{definition}{Definition}
\newtheorem{corollary}{Corollary}
\newtheorem{example}{Example}

\newcommand{\m}[1]{\textcolor{green!50!black}{$\boxed{#1}$}}

\maketitle

\begin{abstract}
In the {\sc course allocation} problem, there are a set of students and a set of courses at a given university. University courses may have different numbers of credits, typically related to different numbers of learning hours, and there may be other constraints such as courses running concurrently. Students may have upper bounds on the number of credits they may be assigned, and courses may have upper quotas on the number of students that can be enrolled. Our goal is to allocate the students to the courses, taking these constraints into account, such that the resulting matching is stable, which means that no student and course(s) have an incentive to break away from the matching and become assigned to one another. We study several definitions of stability and for each we give a mixture of polynomial-time algorithms and hardness results for problems involving verifying the stability of a matching, finding a stable matching or determining that none exists, and finding a maximum size stable matching. We also study variants of the problem with master lists of students, and lower quotas on the number of students allocated to a course, establishing additional complexity results in these settings.
\end{abstract}

\noindent
{\bf Keywords:} Course Allocation problem; credits; master list; lower quotas; stable matching
\section{Introduction}
\label{sec:intro}
\subsection{Background and motivation}
\label{subsec:motivation}
Every year, undergraduate students must select a list of courses to enrol on for the next academic session. Their selections are typically constrained by several factors, mainly that university courses have a limited capacity and that universities enforce a limit on the academic load of students. This load is usually expressed in terms of the total number of course \emph{credits}, with longer courses that have more associated learning hours typically having more credits than shorter ones.

These constraints, along with others such as prerequisites or excluded combinations (e.g., due to concurrent timetabling), have led to a rich variety of papers on methods for finding an optimal (with respect to some given criteria) allocation of students to courses \cite{budish_combinatorial_2011,budish_course_2017,budish_multi-unit_2012,shriya_approaches_2026,sonmez_course_2010}. In general, these models assume that students have preferences over bundles of courses, but that courses are indifferent to the students they are allocated. However, this is not the case in many universities, with some giving a higher priority to senior students, or students with better grades \cite{diebold_matching_2017, elte_regulations_2025}\citep[][page 19, footnote 17]{kornbluth_undergraduate_2024}. An allocation that does not respect these priorities involving students, which can be modelled via course preferences, could lead to \emph{justified envy}, where a student $s$ who is not already allocated to a course $c$ would rather become allocated to $c$ than remain with their existing assignment, and $s$ has a higher priority at $c$ than some other student already assigned to $c$, or \emph{wastefulness}, where $s$ would rather be allocated to $c$ than remain with their current assignment, and $c$ has a vacancy. Such a scenario could lead to student $s$ making a complaint, or even persuading the coordinator of course $c$ to change the allocation. This second possibility could lead to further such rearrangements, resulting in the matching unravelling \cite{roth_1984_evolution}. Hence, we are interested in finding \emph{stable matchings}, which, informally, are allocations in which there is no justified envy or wastefulness.

Student preferences over bundles of courses can allow constraints such as corequisites and scheduling conflicts to be modelled easily. However, it can be challenging to elicit preferences from students over bundles of courses \cite{budish_course_2017, soumalias_machine_2024} and the number of possible bundles is in general exponential in the number of courses. Therefore, in many models \cite{atef_yekta_optimization-based_2020, biswas_2023_algorithmic, cechlarova_pareto_2018, diebold_matching_2017, kornbluth_undergraduate_2024, utture_student_2019}, students have preferences over individual courses, and this is the assumption that we make in this paper. Due to courses having different numbers of credits, stability definitions should be adapted accordingly, and we consider four possible definitions of stability here. All definitions are based on the absence of a blocking pair/coalition, which involves a student and either a single course or a set of courses. In the first definition, \emph{pair stability}, a student may block a matching (i.e., form a blocking pair) with a single course only, as long as no constraints are violated by the switch. In the second one, \emph{pair-size stability}, a student may also block with a single course, as long as she does not reduce her total number of credits by making this switch. In the last two definitions, \emph{coalition stability} and \emph{first-coalition stability}, each student may block with a coalition of courses, as long as she does not reduce her total number of credits. The difference between them is that, in coalition stability, the student has to prefer every course in the coalition over every course that she wants to drop from the matching, while in first-coalition stability, the student only needs to prefer her favourite course from the coalition over her favourite course among those she wants to drop. For both of these coalition-based definitions, each course in the blocking coalition must have an incentive to take on the student in the blocking coalition. For all stability definitions, we study how to verify the stability of a matching and how to find a stable matching, if one exists, and obtain both hardness results and polynomial-time algorithms.

In some cases, all courses within a faculty have the same preferences over students \cite{elte_regulations_2025}. For example, all courses may give priority to students with higher grades over those with lower grades. We show that, when such a uniform ranking of students exists, known as a master list, it is always possible to find a stable matching in polynomial time for all four stability criteria. However, stable matchings can have different sizes, and we show that the problem of finding a stable matching of maximum size is NP-hard, except in the case of pair stability. Finally, we consider the case of courses having lower quotas \cite{biro_college_2010}, a minimum number of students that a course needs in order to run. In general we find that problems involving lower quotas are NP-hard.

\subsection{Related work}
\label{subsec:relatedWork}
The {\sc course allocation} problem has been extensively studied in the past, with papers in this area studying a range of different models of the problem. One of the best known mechanisms for addressing this problem is the Approximate Competitive Equilibrium from Equal Incomes mechanism (A-CEEI) \cite{budish_combinatorial_2011}, in which students rank possible schedules (each schedule is a subset of courses which does not violate any constraint) and each is given an approximately equal budget to bid on courses. The outcome is approximately Pareto efficient and satisfies some other fairness criteria. It was later refined, resulting in the Course Match mechanism \cite{budish_course_2017}. Due to the difficulty that students may have with expressing an ordering across all possible schedules, later work has modified Course Match to use machine learning to better estimate the true preferences of students \cite{soumalias_machine_2024}.

Alternative approaches to the {\sc course allocation} problem involve students ranking courses individually, instead of ranking bundles, but ensuring that the final allocation respects constraints such as the maximum number of credits a student can take or disallowing excluded combinations, such as courses that run concurrently \cite{biswas_2023_algorithmic}. Other types of constraints have also been considered, for example where courses have prerequisites, corequisites or similar constraints that force students, if they desire to enrol in a given course, to also enrol in a different one \cite{cechlarova_pareto_2018}.

Many papers on {\sc course allocation} consider courses to be objects that do not have preferences or priorities over the students that are allocated to them. However, for many university courses around the world, certain groups of students are given a higher priority than others \cite{diebold_matching_2017, elte_regulations_2025}\citep[][page 19, footnote 17]{kornbluth_undergraduate_2024}. This is also the case for the School of Law at the University of Glasgow \cite{glasgowlaw_communication_2025}, whose course allocation model inspired this work. In particular, when courses have preferences, it is critical that the allocation of students to courses is a \emph{stable matching}, because stability is a crucial criterion for the success of a matching mechanism in a bipartite matching market with two-sided preferences \cite{roth_1984_evolution, roth_deferred_2008}.

Stable matching problems were introduced by Gale and Shapley \cite{gale_college_1962}. They considered the {\sc stable marriage} problem ({\sc sm}), a problem in which agents are divided into men and women and must be paired together. They showed that a stable matching, a matching in which no man-woman pair wants to break with their partners and be together instead, always exists and can be found in polynomial time. The algorithm that they presented is known as the Deferred Acceptance mechanism (DA). The {\sc hospitals/residents} problem ({\sc hr}) is a variation of {\sc sm} in which resident doctors apply to hospitals, and hospitals may be assigned more than one resident. As in the case of {\sc sm}, a stable matching always exists and can be found in polynomial time using DA \cite{gale_college_1962}. However, a variant known as {\sc hospitals/residents with couples} ({\sc hrc}), in which some residents form a couple and submit a joint preference list, is NP-hard \cite{ng_1988_complexity, ron_1990_np}, even if each member of every couple has the same preference list as their partner \cite{mcdermid_keeping_2010}. This variant is a special case of {\sc hospitals/residents with sizes} ({\sc hrs}), a variant of {\sc hr} where each resident occupies multiple posts at a hospital, because each couple can be treated as a resident of size $2$ \cite{mcdermid_keeping_2010}. In the case of {\sc course allocation}, if every course has the same number of credits, DA can be used to find a stable matching. However, as in the case of couples in {\sc hrc}, having courses with different numbers of credits means that DA is not guaranteed to output a stable matching.

There has been some previous work on {\sc course allocation} that considers stability. For example, Kornbluth and Kushnir \cite{kornbluth_undergraduate_2024} introduced a competitive equilibrium-based mechanism where the outcome is approximately stable. However, in their model all courses are assumed to have the same academic load, or number of credits, which is not always the case \cite{biswas_2023_algorithmic, budish_course_2017, glasgow_catalogue_2025}. Utture et al.\ \cite{utture_student_2019} do consider a setting where courses have different numbers of credits and other constraints may be present, such as excluding combinations of courses that run in the same time slot, but the algorithmic framework they present does not always output a stable matching. Diebold and Bichler \cite{diebold_matching_2017} also study {\sc course allocation} with stability, but in their model students can be matched to at most one course. Hence, this work is the first 
one that studies {\sc course allocation} with stability in a model where courses have different numbers of credits, there are excluded combinations of courses, and where certain other constraints may be present.

In many matching markets, including some {\sc course allocation} settings \cite{elte_regulations_2025}, agents in one side are ordered according to a master list, such that the preferences of the agents on the other side are consistent with the master list. The presence of a master list guarantees the existence of a stable matching in problems such as {\sc stable roommates} \cite{abraham_stable_2008} (a non-bipartite generalisation of {\sc sm}) and makes some NP-hard problems solvable in polynomial time \cite{abraham_stable_2008}. However, this is not always the case. Generally, problems such as finding a maximum size stable matching are NP-hard even in the presence of master lists, if there are ties present \cite{abraham_stable_2008,irving_stable_2008}.

In some applications, courses have lower quotas, representing a minimum number of students that need to be enrolled in the course in order for the course to run \cite{cechlarova_pareto_2017}, but this extension has not received a great deal of attention in the literature on {\sc course allocation}. However, there are many results for problems such as {\sc hr} (sometimes in the equivalent setting of university admissions). Some papers study the case where every hospital \emph{must} satisfy its lower quota; this case is known as the \emph{no closures} model. The variant where hospitals have the choice of satisfying the lower quota or being closed and not admitting any residents is known as the \emph{closures} model. It is possible to determine in polynomial time whether there exists a stable matching in the {\sc hr} with lower quotas and no closures model, because in every stable matching each hospital has the same number of residents \cite{gale_1985_some, roth_1984_evolution, roth_1986_allocation}, so one simply needs to check whether in one stable matching all hospitals satisfy their lower quotas. On the other hand, if there are ties in the preference lists of the participants, giving the {\sc hospitals/residents problem with ties} ({\sc hrt}), then in the presence of lower quotas and in the no closures model, the problem of deciding whether a stable matching exists is NP-complete.  This follows because determining the existence of a complete stable matching is NP-complete \cite{iwama_2002_stable, manlove_2002_hard} even in the 1-1 restriction of {\sc hrt} known as {\sc smti} ({\sc sm} with ties and incomplete lists), and a complete stable matching can be enforced by setting every lower quota to be $1$.
However, for {\sc hr} with lower quotas and strictly ordered preferences in the closures model, the problem of deciding whether a stable matching exists is NP-complete. \cite{biro_college_2009, biro_college_2010}.

Our model of {\sc course allocation} can be placed in the context of the wider literature on matching problems with preferences \cite{manlove_algorithmics_2013}. For example, if courses are assumed to not have preferences over students, then we obtain a model similar to some previously existing {\sc course allocation} models \cite{biswas_2023_algorithmic, budish_course_2017}, or to {\sc capacitated house allocation} \cite{manlove_algorithmics_2013}. If both students and courses have preferences, but students can be matched with at most one course, we obtain the {\sc hr} problem \cite{gale_college_1962, mcdermid_keeping_2010}. If on the other hand courses can be matched with at most one student, every course has one or two credits, and neither students nor courses have preferences, we obtain the {\sc matching with couples} model \cite{biro_matching_2014}. In this model finding a maximum size matching is NP-hard \cite{biro_matching_2014}.

In the literature on matching under preferences, there have been several papers considering problems where agents have different sizes. For example, Balasundaram et al.\ \cite{balasundaram_stability_2025} studied {\sc hrs}; we will return to this paper in Section \ref{subsec:contributions}. Delacrétaz \cite{delacretaz_2019_stability} considered a version of matching with different sizes in which agents have size $1$ or $2$, and presented different relaxations of stability that ensure the existence of a stable matching. One of them, size-stability, is equivalent to pair-size stability.

Aziz et al.\ \cite{aziz_stability_2018} analysed a model similar to ours in the context of allocating refugee families to host localities, which is similar to our model as refugee families can be regarded as agents with different sizes. Some of their stability definitions are similar to ours; we will elaborate on this point later, in Section \ref{sec:model}. Delacrétaz et al.\ \cite{delacretaz_2023_matching} presented algorithms for finding Pareto optimal, individually rational or interference-free (a concept related to envy-freeness) matchings in the same context. Another model of interest is due to Hoyer and Stroh-Maraun in the context of the {\sc school choice} problem \cite{hoyer_stability_2020}. In their model some students may need more resources than others, which is similar to the fact that, in our model, some courses have a higher number of credits than others. They presented a (not necessarily polynomial-time) algorithm that finds a stable matching if one exists. Another paper studying the same problem uses a variation of the Top-Trading Cycles algorithm to output a Pareto optimal matching \cite{stroh-maraun_weighted_2024}.

\subsection{Our contributions}
\label{subsec:contributions}
In this paper we introduce a new model for the {\sc course allocation} problem, focused on courses with different numbers of credits and preferences over students, in which the goal is to find a matching that is stable, or a stable matching that is of maximum size. We also consider these problems in the presence of lower quotas. Table \ref{tab:resultsTable} presents a summary of our results. Each problem can be identified as follows. Course allocation under pair, first-coalition, coalition and pair-size stability is identified as {\sc ca-p}, {\sc ca-fc}, {\sc ca-c} and {\sc ca-ps} respectively. For the problem of testing whether a given matching is stable, we append {\sc -test}. The problem of finding a stable matching or reporting that none exists has the suffix {\sc -find}. In the presence of a master list of students, we insert {\sc -ml} into the suffix. The problem of finding a maximum size stable matching is represented by {\sc -max}. In the presence of lower quotas, there are two variants: either (i) every course must be assigned at least as many students as its lower quota, in which case there are no course closures (represented by {\sc -lq-nc}) or (ii) a course can be closed (and not run) if it does not meet its lower quota (represented by {\sc lq-cl}). The central results in this paper are those in the {\sc -find}, {\sc -ml-find} and {\sc -ml-max} rows.

\begin{table}[h]\small
\centering
\begin{tabular}{|r|c|c|c|c|}
\cline{2-5}
\multicolumn{1}{r|}{} & pair & first-coalition & coalition & pair-size \\
\multicolumn{1}{r|}{} & {\sc ca-p} & {\sc ca-fc} & {\sc ca-c} & {\sc ca-ps} \\
\cline{1-5}
& & & & \vspace{-4mm} \\
Testing ({\sc -test}) & P [O\ref{obs:PairStabilityVerifying}] & co-NP-c [T\ref{thm:HardnessVerifyingStability}]   & co-NP-c [T\ref{thm:HardnessVerifyingStability}]  & co-NP-c [T\ref{thm:HardnessVerifyingStability}]\\
& & & & \vspace{-4mm} \\
\cline{1-5}
& & & & \vspace{-4mm} \\
Finding ({\sc -find}) & NP-h, $\nexists$ [T\ref{thm:pairStabilityNPCompleteness}]  & NP-h, $\nexists$ [T\ref{thm:firstCoalitionNPCompleteness}]   & Open   & P, $\exists$ [T\ref{thm:pairSizeStabilityFindingStable}] \\
& & & & \vspace{-4mm} \\
\cline{1-5}
Finding, master list & \multirow{2}{*}{P, $\exists$ [P\ref{prop:findingMasterList}]} & \multirow{2}{*}{P, $\exists$ [P\ref{prop:findingMasterList}]} & \multirow{2}{*}{P, $\exists$ [P\ref{prop:findingMasterList}]} & \multirow{2}{*}{P, $\exists$ [P\ref{prop:findingMasterList}]} \\
({\sc -ml-find}) & & & & \\
\cline{1-5}
Max size, master list & \multirow{2}{*}{P [O\ref{obs:masterListUniquePairStable}]} & \multirow{2}{*}{NP-h [T\ref{thm:firstCoalitionStabilityMaxSize}]} & \multirow{2}{*}{NP-h [T\ref{thm:coalitionStabilityMaxSize}]} & \multirow{2}{*}{NP-h [T\ref{thm:pairSizeStabilityMaxSize}]} \\
({\sc -ml-max}) & & & & \\
\cline{1-5}
Lower quotas, master list, & \multirow{2}{*}{P, $\nexists$ [O\ref{obs:lowerQuotasNoClosuresPairStability}]} & \multirow{2}{*}{NP-h, $\nexists$ [O\ref{obs:lowerQuotasNoClosuresThreeDefinitionsStability}]} & \multirow{2}{*}{NP-h, $\nexists$ [O\ref{obs:lowerQuotasNoClosuresThreeDefinitionsStability}]} & \multirow{2}{*}{NP-h, $\nexists$ [O\ref{obs:lowerQuotasNoClosuresThreeDefinitionsStability}]} \\
no closures ({\sc -ml-lq-nc-find}) & & & & \\
\cline{1-5}
Lower quotas, master list, & \multirow{2}{*}{NP-h, $\exists$ [T\ref{thm:lowerQuotasClosuresPairStability}]} & \multirow{2}{*}{NP-h, $\exists$ [O\ref{obs:lowerQuotasClosuresThreeDefinitionsStability}]} & \multirow{2}{*}{NP-h, $\exists$ [O\ref{obs:lowerQuotasClosuresThreeDefinitionsStability}]} & \multirow{2}{*}{NP-h, $\exists$ [O\ref{obs:lowerQuotasClosuresThreeDefinitionsStability}]} \\
closures ({\sc -ml-lq-cl-max}) & & & & \\
\cline{1-5}
\end{tabular}
\caption{Summary of our contributions. The rows represent the different problems of interest, while the columns refer to the different definitions of stability. The symbol $\exists$ indicates that a matching of this type always exists, while $\nexists$ indicates that a matching of this type need not always exist. The complexity results included in this paper are denoted by `P' for polynomial-time solvability, `NP-h' for NP-hardness and `co-NP-c' for co-NP-completeness.  They are also followed by a letter, indicating whether the result is an observation (O), proposition (P) or theorem (T), followed by the observation, proposition or theorem number.}
\label{tab:resultsTable}
\end{table}

Independently of our results, Balasundaram et al.\ \cite{balasundaram_stability_2025} recently 
considered {\sc hrs} under two definitions of stability: firstly \emph{stability}, which is equivalent to pair stability in our model, and secondly \emph{occupancy-stability}, which is equivalent to pair-size stability. They showed that an occupancy-stable matching is guaranteed to exist, and can be found in polynomial time, which is equivalent to Theorem \ref{thm:pairSizeStabilityFindingStable}, and proved that it is NP-hard to find a maximum size occupancy-stable matching, a similar result to Theorem \ref{thm:pairSizeStabilityMaxSize}, although in our theorem students and courses' preferences are derived from a master list, while in their result there is no master list. Balasundaram et al.\ \cite{balasundaram_stability_2025} did not consider constraints such as excluded combinations.

On the other hand, they provide results that do not appear in our work, such as an approximation algorithm for finding a maximum size occupancy-stable (pair-size-stable) matching with performance guarantee strictly better than $3$, and a tighter reduction proving that it is NP-complete to determine whether an instance of {\sc hrs} has a stable (pair-stable) matching. Finally, they also consider the problem with master lists. Although their definition of a master list is not equivalent to ours, they also show that a stable matching always exists, for any definition of stability, as we do in Proposition \ref{prop:findingMasterList}. Hence, their paper \cite{balasundaram_stability_2025} and the preliminary version of our paper \cite{rodriguez_course_2025}, published independently of each other and in short succession, have partially overlapping results, but each considers slightly different problems and contains several original results not present in the other.

\subsection{Structure of the paper}
\label{subsec:structure}
The remainder of this paper is structured as follows. Section \ref{sec:model} defines the problem model formally. Section \ref{sec:verifying} presents results on verifying the stability of a given matching, while Section \ref{sec:finding} contains results related to finding a stable matching (or reporting that none exists). Section \ref{sec:maximum} covers results on finding a maximum size stable matching (both in the absence of, and in the presence of, lower quotas), and Section \ref{sec:conclusion} summarises the results and presents open problems.

\section{Model}
\label{sec:model}
In this section we begin by defining an instance of the {\sc course allocation} ({\sc ca}) problem, which will be our base problem model. An instance $I$ of {\sc ca} involves a set $S$ of students and a set $C$ of courses. Each student $s_i \in S$ has a preference list denoted by $P(s_i) \subseteq C$. Similarly, each course $c_j\in C$ has a preference list denoted by $P(c_j) \subseteq S$. We will abuse notation by referring to preference lists as both a set and an ordered list. If $c_j$ appears in $s_i$'s preference list, we say that $c_j$ is \emph{acceptable} to $s_i$, and vice versa. We assume that a student is acceptable to a course if and only if that course is acceptable to the student. Hence, we say that a student-course pair $(s_i, c_j)$ is \emph{acceptable} if $s_i$ and $c_j$ are acceptable to each other. Given two courses $c_j,c_k \in P(s_i)$, we write $c_j \succ_{s_i} c_k$ if $s_i$ prefers $c_j$ to $c_k$. We define $\succ_{c_j}$ analogously for all $c_j \in C$.

A special case of {\sc ca} called {\sc course allocation with master list} ({\sc ca-ml}) occurs when there is a uniform ranking $\succ_{MLS}$ of all students, known as a \emph{master list}, where $s_i\succ_{MLS} s_j$ if $s_i$ is placed higher than $s_j$ in the uniform ranking. In the presence of a master list, for all $c_j \in C$, $c_j$ respects the master list in the sense that, for any two students $s_i,s_k\in P(c_j)$, $s_i \succ_{c_j} s_k$ if and only if $s_i\succ_{MLS} s_k$.
Throughout the paper, we assume that both individual preference lists and master lists are strict.

We assume that each course $c_j\in C$ has an associated number of \emph{credits}, denoted by $O(c_j)$. We extend this notation to a subset $D\subseteq C$ of courses, i.e., $O(D) = \sum_{c_j \in D} O(c_j)$ is the sum of the number of credits of all courses in $D$. Also course $c_j$ has an \emph{upper quota}, denoted by $q^+(c_j)$, which gives the maximum number of students that can be assigned to $c_j$.
Each student $s_i$ has a maximum number of credits, denoted $T(s_i)$, that she can be assigned from the courses she is allocated.

We now define the concept of a \emph{matching} in instances of {\sc ca}.
\begin{definition}
Given an instance of {\sc ca}, an \emph{assignment} $M$ is a subset of acceptable student/course pairs. For a given student $s_i\in S$, we let $C_M(s_i) = \{c_j \in C : (s_i, c_j) \in M\}$ be the set of courses that $s_i$ is assigned in $M$, and analogously, given a course $c_j\in C$, we let $S_M(c_j) = \{s_i \in S : (s_i, c_j) \in M\}$ be the set of students assigned to $c_j$ in $M$. The assignment $M$ is a \emph{matching} if additionally $M$ satisfies the following conditions:
\begin{itemize}
    \item for each student $s_i\in S$, $O(C_M(s_i)) \leq T(s_i)$, that is, $s_i$ cannot take more credits than her credit limit;
    \item for each course $c_j\in C$, $|S_M(c_j)| \leq q^+(c_j)$, that is, $c_j$ cannot enrol more students than its upper quota.
\end{itemize}
\end{definition}

Given a matching $M$, if $C_M(s_i)=\emptyset$ for a student $s_i\in S$, we will say that $s_i$ is \emph{unmatched} in $M$.  Similarly if $S_M(c_j)=\emptyset$ for a course $c_j\in C$, we will say that $c_j$ is \emph{unmatched} in $M$.  If $|S_M(c_j)|<q^+(c_j)$ for a course $c_j\in C$, we will say that $c_j$ is \emph{undersubscribed} in $M$.

Although in this work we are mostly concerned with a setting where students take any combination of acceptable courses that does not exceed their credit limit, in real life there are other factors that affect the feasibility of a matching. For example, two courses may run at the same time (they are a so-called \emph{excluded combination}), students may have credit limits per quarter or semester, or at most one out of three courses from a given group may be selected, because their content is very similar. In general, such constraints can be modelled by defining the \emph{feasible} subsets of courses for each student. An important special case of such subsets are \emph{downward-feasible subsets}, as defined by Utture et al.\ \cite{utture_student_2019} which we now define here.

\begin{definition}[\cite{utture_student_2019}]
For each $s_i \in S$, a set of \emph{feasible constraints} is a subset $F_i$ of the power set of $P(s_i)$. $F_i$ is \emph{downward-feasible} if $D \in F_i$ and $D' \subseteq D$ implies that $D' \in F_i$. A matching $M$ is \emph{feasible} if, for every $s_i$, $C_M(s_i) \in F_i$. Given a matching $M$ and a student $s_i$, we say that a course $c_j$ is \emph{feasible} for $s_i$ if $c_j\in P(s_i)$ and $C_M(s_i)\cup \{c_j\}\in F_i$.
\end{definition}

That is, downward feasibility means that if it is feasible for a student to be matched with a set of courses, then it is feasible to be matched with a subset of these courses. From the definition it is also clear that credit limits of students are a type of downward-feasible constraints, but we will treat credit limits separately. This is because credit limits are more central to the {\sc ca} problem; almost all universities and {\sc ca} models have credit limits on students, while other downward-feasible constraints do not always apply. Every hardness result in this paper assumes that there are no downward-feasible constraints, and so they still apply if these are present.

Throughout this paper we will assume that we can check in polynomial time whether a given matching is feasible, and similarly whether a set of feasible constraints is downward-feasible. In general downward-feasible constraints can model many rules that can be represented efficiently (e.g., a student must take at most $60$ credits per semester, or no two courses scheduled at the same time can be assigned, etc.)

We define the \emph{size} of a matching to be the total number of credits taken by all students, or alternatively, the number of students enrolled in courses multiplied by their number of credits. That is, the size of a matching $M$ is $\sum_{s_i \in S} \sum_{c_j \in C_M(s_i)} O(c_j) = \sum_{c_j \in C} O(c_j) \cdot |S_M(c_j)|$. We say that a matching $M$ is \emph{course-complete} if every course is filled to its upper quota, that is, for all $c_j \in C$, $|S_M(c_j)| = q^+(c_j)$. A matching is \emph{student-complete} if every student takes as many credits as possible, that is, for all $s_i \in S$, $O(C_M(s_i)) = T(s_i)$. The following observation is then immediate.

\begin{observation}
Each of a course-complete and a student-complete matching is a maximum size matching.
\end{observation}

Our interest lies in finding a stable matching. Informally a stable matching ensures that there is no student and course (or set of courses) that have an incentive to deviate by becoming matched together, possibly disregarding one or more of their existing assignments.

Because courses have different numbers of credits, there are several possible definitions of stability. The most natural definition is the following, which is analogous to the standard definition of stability in 
{\sc sm} and {\sc hr}.
It is equivalent to the definitions of stability in {\sc hrs} \cite{balasundaram_stability_2025, mcdermid_keeping_2010} and in {\sc school choice} papers by Hoyer and Stroh-Maraun \cite{hoyer_stability_2020} and Stroh-Maraun \cite{stroh-maraun_weighted_2024}. It is also similar to the definition of stability due to Aziz et al.\ \cite{aziz_stability_2018}.

\begin{definition}
In a matching $M$, a student $s_i$ and a course $c_j$ form a \emph{blocking pair} if $(s_i, c_j)$ is acceptable, $(s_i, c_j) \notin M$, and the following conditions hold:
\begin{enumerate}
    \item $c_j$ is undersubscribed 
    \emph{or} $c_j$ prefers $s_i$ over some $s_k$ assigned to $c_j$. That is, $|S_M(c_j)| < q^+(c_j)$ \emph{or} $s_i \succ_{c_j} s_k$ for some $s_k \in S_M(c_j)$.
    \item There exists a subset $D_i\subseteq C_M(s_i)$ (possibly empty) such that $s_i$ prefers $c_j$ to every $c_k\in D_i$, and $O(C_M(s_i)) - O(D_i) + O(c_j) \leq T(s_i)$. That is, $s_i$ has enough credits remaining for $c_j$, possibly after removing a subset of courses that $s_i$ prefers $c_j$ to.
    \item If there are downward-feasible constraints, $(C_M(s_i)\setminus D_i)\cup \{c_j\}\in F_i$.  That is, removing $D_i$ and adding $c_j$ to $C_M(s_i)$ does not result in an infeasible matching.
\end{enumerate}
If $M$ has no blocking pair then it is \emph{pair-stable}. 
{\sc ca} with pair stability is denoted {\sc ca-p}.
\end{definition}

An example of a matching with a blocking pair is given in Figure \ref{fig:ExampleCourseAllocationPairStability}. However, we can see that in this example $s_1$ would have fewer credits if she was matched with $c_1$ instead of $c_3$. In such a situation, it may be argued that such a student is not better off, since students generally need to have a curriculum that is as complete as possible. That motivates our next definition of stability. This definition is similar to the definition of \emph{stability by demand} due to Aziz et al.\ \cite{aziz_stability_2018}, and equivalent to the \emph{occupancy-stability} definition due to Balasundaram et al. \cite{balasundaram_stability_2025} and the \emph{size-stability} definition due to Delacrétaz \cite{delacretaz_2019_stability}.

\begin{figure}[ht]
    \centering
    \begin{tabular}{|r|c|l|c|r|c|c|l|}
        \cline{1-3} \cline{5-8}
        & & & & & & & \vspace{-2mm}\\
         $s_i$ & $T(s_i)$ & $\succ_{s_i}$ & \hspace{1mm} & $c_j$ & $O(c_j)$ & $q^+(c_j)$ & $\succ_{c_j}$ \\
        & & & & & & & \vspace{-2mm}\\
        \cline{1-3} \cline{5-8}
        & & & & & & & \vspace{-2mm}\\
        
        $s_1$ & 2 & {\Circled{$c_1$}} $\succ$ {\m{c_3}} $\succ$ $c_2$ & & 
        
        $c_1$ & 1 & 1 & $s_2$ $\succ$ \Circled{$s_1$}
        \vspace{-2mm}\\
        & & & & & & & \\
        
        $s_2$ & 1 & \m{c_2} $\succ$ $c_1$ & & 
        
        $c_2$ & 1 & 1 & $s_1$ $\succ$ \m{s_2} \vspace{-2mm}\\
        & & & & & & & \\
        \cline{1-3} \multicolumn{4}{c|}{} & 
        
        $c_3$ & 2 & 1 & \m{s_1} \vspace{-2mm}\\
        \multicolumn{4}{c|}{} & & & & \\
        \cline{5-8}
    \end{tabular}
    \caption{Example of a matching in a {\sc ca} instance, highlighted in green, and a blocking pair, highlighted in red.}
    \label{fig:ExampleCourseAllocationPairStability}
\end{figure}

\begin{definition}
In a matching $M$, a student $s_i$ and a course $c_j$ form a \emph{size-blocking pair} if $(s_i, c_j)$ is acceptable, $(s_i, c_j) \notin M$, and the following conditions hold:
\begin{enumerate}
    \item Either $c_j$ is undersubscribed 
    \emph{or} $c_j$ prefers $s_i$ over some $s_k$ assigned to $c_j$. That is, $|S_M(c_j)| < q^+(c_j)$ \emph{or} $s_i \succ_{c_j} s_k$ for some $s_k \in S_M(c_j)$.
    \item There exists a subset $D_i\subseteq C_M(s_i)$ (possibly empty) such that $s_i$ prefers $c_j$ to every $c_k\in D_i$, $O(c_j) \geq O(D_i)$, and $O(C_M(s_i)) - O(D_i) + O(c_j) \leq T(s_i)$. That is, $s_i$ has enough credits remaining for $c_j$, possibly after removing a subset of courses that $s_i$ prefers $c_j$ to and with a total number of credits less than or equal to that of $c_j$.
    \item If there are downward-feasible constraints, $(C_M(s_i)\setminus D_i)\cup \{c_j\}\in F_i$.
    That is, removing $D_i$ and adding $c_j$ to $C_M(s_i)$ does not result in an infeasible matching.
\end{enumerate}
If $M$ has no size-blocking pair then it is \emph{pair-size-stable}. {\sc ca} with pair-size stability is denoted {\sc ca-ps}.
\end{definition}

The only difference between pair stability and pair-size stability is that we require that a student in a size-blocking pair must not end up with fewer credits than they started with when participating in a size-blocking pair. However, this may have unintended consequences. For example, consider student $s_1$ from Figure \ref{fig:ExampleCourseAllocationPairSizeStability}. If $s_1$ is matched with $c_3$, then she would prefer to be matched with $c_1$ and $c_2$ instead. However, because $O(c_3) = 2$ and $O(c_1) = O(c_2) = 1$, $s_1$ cannot form a size-blocking pair with either $c_1$ or $c_2$. This motivates our next stability definition, in which we allow students to block with more than one course.

\begin{figure}[ht]
    \centering
    \begin{tabular}{|r|c|l|c|r|c|c|l|}
        \cline{1-3} \cline{5-8}
        & & & & & & & \vspace{-2mm}\\
         $s_i$ & $T(s_i)$ & $\succ_{s_i}$ & \hspace{1mm} & $c_j$ & $O(c_j)$ & $q^+(c_j)$ & $\succ_{c_j}$ \\
        & & & & & & & \vspace{-2mm}\\
        \cline{1-3} \cline{5-8}
        & & & & & & & \vspace{-2mm}\\
        
        $s_1$ & 2 & {\Circled{$c_1$}} $\succ$ {\Circled{$c_2$}} $\succ$ \m{c_3} & & 
        
        $c_1$ & 1 & 1 & $s_2$ $\succ$ \Circled{$s_1$}
        \vspace{-2mm}\\
        & & & & & & & \\
        
        $s_2$ & 1 & \m{c_2} $\succ$ $c_1$ & & 
        
        $c_2$ & 1 & 1 & {\Circled{$s_1$}} $\succ$ \m{s_2} \vspace{-2mm}\\
        & & & & & & & \\
        \cline{1-3} \multicolumn{4}{c|}{} & 
        
        $c_3$ & 2 & 1 & \m{s_1} \vspace{-2mm}\\
        \multicolumn{4}{c|}{} & & & & \\
        \cline{5-8}
    \end{tabular}
    \caption{Example of a pair-size-stable matching in a {\sc ca} instance, highlighted in green. A subset of courses that $s_1$ may prefer over $c_3$ is shown in red.}
    \label{fig:ExampleCourseAllocationPairSizeStability}
\end{figure}

\begin{definition}
In a matching $M$, a student $s_i$ and a set of courses $B \subseteq C$ form a \emph{blocking coalition} if for all $c_j \in B$, $(s_i, c_j)$ is acceptable, for all $c_j \in B$, $(s_i, c_j) \notin M$, and the following conditions hold:
\begin{enumerate}
    \item for each $c_j\in B$, either $c_j$ is undersubscribed
    \emph{or} $c_j$ prefers $s_i$ over some $s_k$ assigned to $c_j$. That is, for every $c_j \in B$, either $|S_M(c_j)| < q^+(c_j)$ \emph{or} $s_i \succ_{c_j} s_k$ for some $s_k \in S_M(c_j)$.
    \item There exists a subset $D_i\subseteq C_M(s_i)$ (possibly empty) such that $s_i$ prefers every $c_j \in B$ to every $c_k\in D_i$, $O(B) \geq O(D_i)$ and $O(C_M(s_i)) - O(D_i) + O(B) \leq T(s_i)$. That is, $s_i$ has enough credits remaining for all courses in $B$, possibly after removing a subset of courses that $s_i$ prefers all courses in $B$ to, and with a total number of credits less than or equal to that of $B$.
    \item If there are downward-feasible constraints, $(C_M(s_i)\setminus D_i)\cup B\in F_i$.    
    That is, removing $D_i$ and adding $B$ to $C_M(s_i)$ does not result in an infeasible matching.
\end{enumerate}
If $M$ has no blocking coalition then it is \emph{coalition-stable}. {\sc ca} with coalition stability is denoted {\sc ca-c}.
\end{definition}

Under coalition stability, $s_1$ from Figure \ref{fig:ExampleCourseAllocationPairSizeStability} can form a blocking coalition with $c_1$ and $c_2$ if she is matched with $c_3$. Moving on, we now consider Figure \ref{fig:ExampleCourseAllocationCoalitionStability}. Here, if $s_1$ is matched with $c_3$, then she cannot form a blocking coalition with $c_1$ and $c_2$, because she prefers $c_3$ over $c_2$. However, it is plausible that $s_1$ would rather be matched with $c_1$ and $c_2$ instead, because her best course in the new matching is strictly better than her best course in the old matching. Hence, we obtain a new definition of stability.

\begin{figure}[ht]
    \centering
    \begin{tabular}{|r|c|l|c|r|c|c|l|}
        \cline{1-3} \cline{5-8}
        & & & & & & & \vspace{-2mm}\\
         $s_i$ & $T(s_i)$ & $\succ_{s_i}$ & \hspace{1mm} & $c_j$ & $O(c_j)$ & $q^+(c_j)$ & $\succ_{c_j}$ \\
        & & & & & & & \vspace{-2mm}\\
        \cline{1-3} \cline{5-8}
        & & & & & & & \vspace{-2mm}\\
        
        $s_1$ & 2 & {\Circled{$c_1$}} $\succ$ {\m{c_3}} $\succ$ \Circled{$c_2$} & & 
        
        $c_1$ & 1 & 1 & $s_2$ $\succ$ \Circled{$s_1$}
        \vspace{-2mm}\\
        & & & & & & & \\
        
        $s_2$ & 1 & \m{c_2} $\succ$ $c_1$ & & 
        
        $c_2$ & 1 & 1 & \Circled{$s_1$} $\succ$ \m{s_2} \vspace{-2mm}\\
        & & & & & & & \\
        \cline{1-3} \multicolumn{4}{c|}{} & 
        
        $c_3$ & 2 & 1 & \m{s_1} \vspace{-2mm}\\
        \multicolumn{4}{c|}{} & & & & \\
        \cline{5-8}
    \end{tabular}
    \caption{Example of a coalition-stable matching in a {\sc ca} instance, highlighted in green. A subset of courses that $s_1$ may prefer over $c_3$ is shown in red.}
    \label{fig:ExampleCourseAllocationCoalitionStability}
\end{figure}

\begin{definition}
In a matching $M$, a student $s_i$ and a set of courses $B \subseteq C$ form a \emph{first-blocking coalition} if for all $c_j \in B$, $(s_i, c_j)$ is acceptable, for all $c_j \in B$, $(s_i, c_j) \notin M$, and the following conditions hold:
\begin{enumerate}
    \item for each $c_j\in B$, either $c_j$ is undersubscribed
    \emph{or} $c_j$ prefers $s_i$ over some $s_k$ assigned to $c_j$. That is, for every $c_j \in B$, either $|S_M(c_j)| < q^+(c_j)$ \emph{or} $s_i \succ_{c_j} s_k$ for some $s_k \in S_M(c_j)$.
    \item There exists a subset $D_i\subseteq C_M(s_i)$ (possibly empty) such that $s_i$ prefers her most-preferred course in $B$ to her most-preferred course in $D_i$, $O(B) \geq O(D_i)$ and $O(C_M(s_i)) - O(D_i) + O(B) \leq T(s_i)$. That is, $s_i$ has enough credits remaining for all courses in $B$, possibly after removing a subset of courses $D_i$ with a total number of credits less than or equal to that of $B$ such that $s_i$ prefers her most-preferred course in $B$ over her most-preferred course in $D_i$.
    \item If there are downward-feasible constraints, $(C_M(s_i)\setminus D_i)\cup B\in F_i$.    
    That is, removing $D_i$ and adding $B$ to $C_M(s_i)$ does not result in an infeasible matching.
\end{enumerate}
If $M$ has no first-blocking coalition then it is \emph{first-coalition-stable}. {\sc ca} with first-coalition stability is denoted {\sc ca-fc}.
\end{definition}

The four stability definitions can be ordered in terms of their strength with the aid of the following proposition.

\begin{proposition}
Pair stability implies first-coalition stability, which implies coalition stability, which implies pair-size stability.
\label{prop:StabilityDefinitions}
\end{proposition}
\begin{proof}
We will prove all three implications by contradiction. First, assume that a matching $M$ is pair-stable but not first-coalition-stable. Hence, there exist a student $s_i$ and a set of courses $B$ such that $s_i$ has enough space for all courses in $B$, possibly after removing a subset of courses $D_i \subseteq C_M(s_i)$ that $s_i$ prefers her most-preferred course in $B$ to. Let this course be $c_j$. Then $(s_i, c_j)$ is a blocking pair, so $M$ is not pair-stable, a contradiction. If there are downward-feasible constraints, $(s_i, c_j)$ is still a blocking pair, because $s_i$ and $B$ form a first-blocking coalition, which means that $(C_M(s_i)\setminus D_i)\cup B\in F_i$. 
Moreover $(C_M(s_i)\setminus D_i)\cup \{c_j\}\subseteq (C_M(s_i)\setminus D_i)\cup B$ and hence $(C_M(s_i)\setminus D_i)\cup \{c_j\}\in F_i$.

Now, assume that $M$ is first-coalition-stable, but not coalition-stable. Then there exist a student $s_i$ and a set of courses $B$ such that $s_i$ has enough space for all courses in $B$, possibly after removing a subset of courses $D_i \subseteq C_M(s_i)$ such that $s_i$ prefers each course in $B$ to each course in $D_i$. This implies that $s_i$ prefers her most-preferred course in $B$ over her most-preferred course in $D_i$, so $(s_i,B$) is a first-blocking coalition, and $M$ is not first-coalition-stable, a contradiction. If there are downward-feasible constraints, $(s_i,B$) is still a first-blocking coalition, since $(s_i,B$) is a blocking coalition, which means that $(C_M(s_i)\setminus D_i)\cup B\in F_i$.

Finally, assume that $M$ is coalition-stable, but not pair-size-stable. Then there exist a student $s_i$ and a course $c_j$ such that $(s_i, c_j)$ is a size-blocking pair. However, this implies that $(s_i,\{c_j\})$ is a blocking coalition, so $M$ is not coalition-stable, a contradiction. If there are downward-feasible constraints, $(s_i,\{c_j\})$ is still a blocking coalition, because $(s_i, c_j)$ is a size-blocking pair, which means that $(C_M(s_i)\setminus D_i)\cup \{c_j\}\in F_i$.
\end{proof}

\section{Verifying stability}
\label{sec:verifying}
In this section we consider the complexity of the problem of determining whether a given matching is stable, with respect to the stability criteria defined in the previous section. Firstly, we remark that it is straightforward to check whether an assignment is a matching, and whether a matching is feasible, in the presence of downward-feasible constraints. One simply needs to check if the assignment respects upper quotas and credit limits, and if it respects all downward-feasible constraints. On the other hand, checking whether a (feasible) matching is stable may be trickier, depending on the stability definition. Let {\sc ca-x-test} be the problem of deciding whether a given matching $M$ is stable, for a given definition of stability x. Under pair stability, one can verify the stability of $M$ in polynomial time. For each course we can compute its worst assigned student in $M$, within an overall time complexity of $O(L)$, where $L$ is the total number of acceptable pairs. For any acceptable student-course pair $(s_i, c_j) \notin M$, we can check if $c_j$ has enough capacity for $s_i$, or if $c_j$ prefers $s_i$ over its worst assigned student in $M$, which can be done in $O(1)$ time, using suitable data structures. Similarly, we can check in $O(m)$ time, where $m$ is the number of courses, if $s_i$ has enough remaining credits for $c_j$, possibly after removing all courses in $C_M(s_i)$ that $s_i$ prefers $c_j$ to. Recall that under this definition a student can end up with fewer credits after satisfying a blocking pair. The overall time complexity is then $O(Lm)$. Thus, we obtain the following observation.

\begin{observation}
{\sc ca-p-test} is solvable in $O(Lm)$ time, where $m = |C|$ and $L$ is the number of acceptable pairs.
\label{obs:PairStabilityVerifying}
\end{observation}

However, verifying stability is hard for the other three definitions of stability, which we will show in Theorem \ref{thm:HardnessVerifyingStability}. The difference is that in a blocking pair (for pair stability), a student is allowed to end up with fewer credits than she had in the matching $M$, but that is not allowed in the other three definitions of stability. This fact is crucial in our transformation from the {\sc subset sum} problem. An instance of {\sc subset sum} consists of a multiset of elements $X$, where each element $e\in X$ has \emph{size} $s(e)$, and a target integer $T$. The goal is to determine the existence of a subset of $X$ where the sum of the sizes of its elements equals $T$. The problem is weakly NP-complete, even if the size of every element in $X$ is a positive integer \citep[page 223]{garey_1979_computers}.

\begin{theorem}
The following problems are weakly co-NP-complete: {\sc ca-fc-test}, {\sc ca-c-test} and {\sc ca-ps-test}.
\label{thm:HardnessVerifyingStability}
\end{theorem}
\begin{proof}
These problems are in co-NP because given a matching $M$, a student $s$, a set of courses $B$ (of size $1$ in the case of pair-size stability) with which $s$ would form a size-blocking pair/blocking coalition/first-blocking coalition, and a set of courses $D$ that $s$ would drop, we can check whether $s$ blocks the matching $M$ with the courses in $B$ in polynomial time.

Consider a {\sc subset sum} instance $I$, with multiset $X$ and target number $T$. We create an instance $J$ of {\sc ca} as follows. For every element $e \in X$, we create a course $c_e$, with upper quota $1$ and $O(c_e) = s(e)$. We create one additional course, $b$, with upper quota $1$ and $O(b) = T$. We also create one student, $s$, with $T(s)$ equal to the sum of the sizes of all elements in $X$. The preference list of $s$ is $b$ first, then all of the other courses in any order. The matching that we want to verify is $M = \{(s, c_e) | e \in X\}$. Thus, the matching includes every acceptable pair except for $(s, b)$. Also, $O(C_M(s)) = T(s)$.

Assume that there exists a subset $X' \subseteq X$ such that $\sum_{e \in X'} s(e) = T$. Let $D = \{c_e | e \in X'\}$ be a subset of courses. Then $(s,b)$ is a size-blocking pair, because $b$ has enough capacity for $s$, $s$ prefers $b$ over every course in $D$, $O(b) = O(D) = T$, and $O(C_M(s)) - O(D) + O(b) = T(s)$. Thus, $M$ is not pair-size-stable. We have that $(s,B)$ is a blocking coalition and a first-blocking coalition, too, so $M$ is not coalition-stable or first-coalition-stable either.

Now, assume that $M$ is not pair-size-stable. Thus, there exists at least one size-blocking pair. The only acceptable pair that is not in $M$ is $(s,b)$, so that is our size-blocking pair. As $(s,b)$ is a size-blocking pair, $s$ prefers $b$ over a subset of courses $D \subseteq C_M(s)$ such that $O(b) \geq O(D)$ and $O(C_M(s)) - O(D) + O(b) \leq T(s)$. As $O(C_M(s)) = T(s)$, we have that $O(b) \leq O(D)$, which combined with $O(b) \geq O(D)$ gives us $O(b) = T = O(D)$. As $b \notin D$, all courses in $D$ are of type $c_e$. Thus, we can build a subset $X' \subseteq X$, such that $X = \{e | c_e \in D\}$ and $\sum_{e \in X'} s(e) = T$. Therefore, there exists a subset of $X$ whose elements' sizes add up to $T$. This direction also works in the case of coalition stability or first-coalition stability.
\end{proof}

However, just as in the case of {\sc subset sum}, a pseudo-polynomial dynamic programming algorithm can verify the stability of a given matching, for either pair-size stability, coalition stability, or first-coalition stability. Thus, if the number of credits of all courses is not exponentially large, verifying the stability of a matching can be done in polynomial time. In real-life settings, it is reasonable to assume that course credits are polynomially bounded.\footnote{For example, all courses taken by Honours students at the School of Law, University of Glasgow \cite{glasgowlaw_communication_2025}, carry 20 or 40 credits.  For the purposes of a dynamic programming algorithm, each course's number of credits, and each student's maximum number of credits, can be divided by 20.}

As an aside, verifying the stability of a matching is weakly co-NP-complete even in an instance of {\sc ca-ml}, where students are ordered in a master list. In the proof of Theorem \ref{thm:HardnessVerifyingStability} there is a single student, who trivially forms a master list on her own.

\section{Finding a stable matching}
\label{sec:finding}
\subsection{Arbitrary course preferences}
\label{subsec:findingAbritrary}

We now move onto the problem of finding a stable matching or reporting that none exists, for our four stability definitions. Let {\sc ca-x-find} be the problem of finding a stable matching in a {\sc ca} instance $I$, for a given definition of stability x, or determining that none exists. As in the case of testing, the complexity of {\sc ca-x-find} depends on the definition of stability. In fact, depending on the definition of stability, a stable matching may or may not be guaranteed to exist.

Let us first consider the problem of finding a pair-size-stable matching. In order to do so, we use the following procedure. First, we group courses by number of credits, from largest to smallest. Let $\{O(c_j) : c_j \in C\} = \{N_1, \dots, N_R\}$ for some $R \geq 1$, where $N_r > N_{r+1}$ ($1 \leq r \leq R-1$). Let $C_r = \{c_j \in C : O(c_j) = N_r\}$ for each $r$ ($1\leq r\leq R$).

A pair-size-stable matching can be found using a variation of Deferred Acceptance, which has $R$ rounds. In the first round, we consider the set of all students $S$ and the subset $C_1$ of courses, and match them using DA. We take this partial matching onto round $2$, in which we consider all students $S$ and the subset $C_2$ of courses. In general, in round $i$, when we consider all students $S$ and the subset $C_i$ of courses, each student applies to courses in $C_i$ from the beginning of their preference list. Importantly, if a student $s_i$ and a course $c_j$ are matched in some round, they remain matched until the end of execution, so at each round the number of credits that $s_i$ has available for courses stays the same or decreases. We continue until $R$ rounds have been executed, at which point the entire set $C = \bigcup_{i=1}^R C_i$ of courses will have been considered. The pseudocode for this algorithm is given in Algorithm \ref{alg:pairSizeArbitrary}.

\begin{algorithm}[h]
\small
\caption{Algorithm for finding a pair-size-stable (feasible) matching}
\label{alg:pairSizeArbitrary}
\begin{algorithmic}[1]
\REQUIRE Course Allocation instance $I$
\ENSURE return a pair-size-stable matching $M$ in $I$
\STATE $M := \emptyset$;
\FOR{$r \in 1 \dots R$}
    \WHILE{some student $s_i\in S$ satisfies $O(C_M(s_i))+N_r\leq T(s_i)$ {\bf and} $s_i$ has not applied to all feasible courses in $C_r$ on her list}
      \STATE $c_j :=$ $s_i$'s most-preferred feasible course in $C_r$ that $s_i$ has not yet applied to;
      \STATE $s_i$ applies to $c_j$;
      \IF{$|S_M(c_j)| < q^+(c_j)$}
        \STATE $M := M\cup \{(s_i,c_j)\}$;
      \ELSE
        \STATE $s_k :=$ worst student assigned to $c_j$;
        \IF{$c_j$ prefers $s_i$ to $s_k$}
          \STATE $M := (M\cup \{(s_i,c_j)\})\backslash \{(s_k,c_j)\}$;
          \STATE $c_j$ rejects $s_k$;
        \ELSE
           \STATE $c_j$ rejects $s_i$;
        \ENDIF
      \ENDIF
    \ENDWHILE
\ENDFOR
\RETURN $M$;\\
\end{algorithmic}
\normalsize
\end{algorithm}

It is clear that Algorithm \ref{alg:pairSizeArbitrary} finds a matching $M$. Furthermore, if there are downward-feasible constraints, $M$ is feasible, because every $s_i$ applies only to courses $c_j$ such that adding $(s_i, c_j)$ to $M$ does not violate downward-feasible constraints. With appropriate data structures, the overall complexity of Algorithm \ref{alg:pairSizeArbitrary} is $O(RL)=O(nm^2)$, where $n = |S|$, $m = |C|$ and $L$ is the number of acceptable pairs. This is because for every student we consider each course in her preference list at most $m$ times, if there are downward-feasible constraints, as courses that were previously unfeasible might become feasible for some $s_i$ if she is rejected from some $c_j$. Meanwhile, when a student applies to a course, the course can perform all necessary steps in constant time, giving a total time complexity of $O(nm^2)$, where $n = |S|$ and $m = |C|$. If there are no downward-feasible constraints then the time complexity is $O(nm)$. We now show that $M$ is pair-size-stable.

\begin{theorem}
Algorithm \ref{alg:pairSizeArbitrary} always finds a pair-size-stable (feasible) matching in $O(nm^2)$ time, where $n = |S|$ and $m = |C|$.
\label{thm:pairSizeStabilityFindingStable}
\end{theorem}
\begin{proof}
We will argue by contradiction. Assume the matching $M$ found by Algorithm \ref{alg:pairSizeArbitrary} is not pair-size-stable, and that there is a size-blocking pair $(s_i, c_j)$, such that $c_j \in C_r$ for some $r$. Thus, either $s_i$ has enough remaining credits to match with $c_j$ (and adding $(s_i, c_j)$ to $M$ results in a feasible matching) or $s_i$ prefers $c_j$ over every course $c_k$ from a subset of courses $D_i$ such that $s_i$ is assigned to every $c_k \in D_i$, $O(c_j) \geq \sum_{c_k \in D_i} O(c_k)$, and, if there are downward-feasible constraints, $(M \cup \{(s_i, c_j)\}) \backslash D_i$ is feasible.

Assume that the first case holds. The course $c_j$ is in subset $C_r$, meaning that it was considered in the $r$-th round. As $s_i$ has enough credits for $c_j$ after the execution of the algorithm, it had enough credits during round $r$, because matches between students and courses at the end of a given round $k$ cannot be broken subsequently, so after every round the number of remaining credits for $s_i$ stays the same or decreases. Thus, $s_i$ applied to $c_j$ during round $r$, so if they are not matched at the end of round $r$ then $c_j$ rejected $s_i$ for a better student. Hence, $(s_i, c_j)$ was not a size-blocking pair at the end of round $r$. Because all matches formed in round $r$ are part of the final pair-size-stable matching, $c_j$ prefers all of its students over $s_i$, so $(s_i, c_j)$ is not a size-blocking pair, a contradiction. Therefore, this case cannot occur.

Now, assume the second case holds, so $s_i$ prefers $c_j$ over every course $c_k$ from a subset of courses $D_i$ such that $s_i$ is assigned in $M$ to every $c_k \in D_i$ and $O(c_j) \geq \sum_{c_k \in D_i} O(c_k)$. As $O(c_j) \geq \sum_{c_k \in D_i} O(c_k)$, $O(c_j) \geq O(c_k)$ $\forall c_k \in D_i$, no $c_k \in D_i$ has a number of credits strictly greater than $O(c_j)$. Additionally, no $c_k \in D_i$ has a number of credits equal to $O(c_j)$, for if there were such a $c_k$ and $(s_i, c_j)$ was a size-blocking pair, then that would imply that in the round where $c_j$ and $c_k$ were considered, Deferred Acceptance did not find a stable matching, which is a contradiction. Therefore, $c_j$ has a number of credits strictly greater than every $c_k \in D_i$. However, this means that, during the algorithm's execution, $c_j$ was considered in an earlier round than every $c_k \in D_i$, say round $r$. During that round, $s_i$ had enough remaining credits for $c_j$, because, after removing every course from $D_i$, it has enough credits for $c_j$ after the execution of the algorithm. Using a similar argument to the first case, $(s_i, c_j)$ was not a size-blocking pair at the end of round $r$, meaning that it is not a size-blocking pair upon termination of the algorithm, a contradiction. Therefore, this case also cannot occur. By assuming either case holds, we arrive at a contradiction, which is caused by the assumption that there exists a size-blocking pair. Therefore, no such pair exists, and the matching is pair-size-stable.
\end{proof}

On the other hand, a first-coalition-stable matching, or pair-stable matching, may not always exist. Consider, for example, the instance in Example \ref{exmp:ExampleNoInstance}.

\begin{example}
Let us look at every possible matching in the instance from Figure \ref{fig:ExampleNoInstanceFigure}. Then for every matching $M$ we can find a first-blocking coalition, which implies that $M$ is not first-coalition-stable, and thus that it is not pair-stable.

\begin{itemize}
    \item If $M = \{(s_1, c_1), (s_1, c_2)\}$, then $(s_2, c_1)$ is a first-blocking coalition.
    \item If $M = \{(s_1, c_1), (s_2, c_2)\}$, then $(s_1, c_2)$ is a first-blocking coalition.
    \item If $M = \{(s_1, c_2), (s_2, c_1)\}$, then $(s_1, c_3)$ is a first-blocking coalition.
    \item If $M = \{(s_1, c_3), (s_2, c_2)\}$, then $(s_1, \{c_1, c_2\})$ is a first-blocking coalition.
    \item If $M = \{(s_1, c_3), (s_2, c_1)\}$, then $(s_2, c_2)$ is a first-blocking coalition.
\end{itemize}
\label{exmp:ExampleNoInstance}
\end{example}

\begin{figure}[h]
    \centering
    \begin{tabular}{|r|c|l|c|r|c|c|l|}
        \cline{1-3} \cline{5-8}
        & & & & & & & \vspace{-2mm}\\
         $s_i$ & $T(s_i)$ & $\succ_{s_i}$ & \hspace{1mm} & $c_j$ & $O(c_j)$ & $q^+(c_j)$ & $\succ_{c_j}$ \\
        & & & & & & & \vspace{-2mm}\\
        \cline{1-3} \cline{5-8}
        & & & & & & & \vspace{-2mm}\\
        
        $s_1$ & 2 & $c_1$ $\succ$ $c_3$ $\succ$ $c_2$ & & 
        
        $c_1$ & 1 & 1 & $s_2$ $\succ$ $s_1$
        \vspace{-2mm}\\
        & & & & & & & \\
        
        $s_2$ & 1 & $c_2$ $\succ$ $c_1$ & & 
        
        $c_2$ & 1 & 1 & $s_1$ $\succ$ $s_2$ \vspace{-2mm}\\
        & & & & & & & \\
        \cline{1-3} \multicolumn{4}{c|}{} & 
        
        $c_3$ & 2 & 1 & $s_1$ \vspace{-2mm}\\
        \multicolumn{4}{c|}{} & & & & \\
        \cline{5-8}
    \end{tabular}
    \caption{Instance of {\sc ca}. This instance is derived from an {\sc hrs} instance due to McDermid and Manlove \cite{mcdermid_keeping_2010}.}
    \label{fig:ExampleNoInstanceFigure}
\end{figure}

Determining whether a pair-stable matching exists is NP-complete, which can be seen by considering a connection between {\sc ca} and {\sc hrs} ({\sc hrs} was introduced in Section \ref{sec:intro}), which we now describe.
With respect to the stability definition given by McDermid and Manlove \cite{mcdermid_keeping_2010}, the problem of deciding whether a stable matching exists, given an instance of {\sc hrs}, is NP-complete. Moreover this result holds even if each resident has size 1 or 2, each hospital has upper quota at most 2, and each preference list is of length at most 3. By replacing residents by courses and hospitals by students, we can then obtain a simple reduction to {\sc ca} where each student $s_i$ satisfies $T(s_i)\leq 2$ and each course $c_j$ satisfies $O(c_j)\in \{1,2\}$ and $q(c_j)=1$. Using this transformation, McDermid and Manlove's stability definition in the {\sc hrs} instance corresponds to pair stability in the transformed {\sc ca} instance. The following result is therefore immediate.

\begin{theorem}[\cite{mcdermid_keeping_2010}]
{\sc ca-p-find} is NP-hard, even if there are no downward-feasible constraints, each student can take at most two credits, each course has one or two credits, each course has upper quota $1$, and each preference list is of length at most $3$.
\label{thm:pairStabilityNPCompleteness}
\end{theorem}

Similarly, determining the existence of a first-coalition-stable matching is NP-hard, and we show that by transforming from {\sc hrs}, though in this case a more complex reduction is required.

\begin{theorem}
{\sc ca-fc-find} is NP-hard, even if there are no downward-feasible constraints, each student can take at most two credits, each course has one or two credits, each course has upper quota $1$, and each preference list is of length at most $4$.
\label{thm:firstCoalitionNPCompleteness}
\end{theorem}
\begin{proof}
We prove this result by transforming from the problem of determining the existence of a stable matching in the {\sc hrs} setting, which was shown to be NP-complete by McDermid and Manlove \cite{mcdermid_keeping_2010}. In the {\sc hrs} instance $I$ that we are transforming from, each resident has size $1$ or $2$, each hospital has upper quota $2$, and the length of every preference list is at most $3$.

The transformation proceeds as follows. Every resident $r_j$ in $I$ becomes a course $c_j$ in the {\sc ca} instance $J$, with upper quota $1$, and number of credits $1$ or $2$, depending on the size of $r_j$. Every hospital $h_i$ becomes a student $s_i$, with $T(s_i) = 2$. We use this bijection to transform every preference list. For example, if resident $r_j$ has preference list $h_i \succ_{r_j} h_k \succ_{r_j} h_l$, now course $c_j$ has preference list $s_i \succ_{c_j} s_k \succ_{c_j} s_l$. Additionally, for every student $s_i$ we create a dummy course $c_i'$, with one credit, and append this course to the end of $s_i$'s preference list. The preference list of $c_i'$ contains $s_i$ only.

Assume that the {\sc hrs} instance $I$ has a stable matching $M$. Then we build a first-coalition-stable matching $M'$ in $J$ as follows. We take $M$ and swap residents and hospitals for courses and students, respectively. Then, for every student $s_i$ such that $O(C_{M'}(s_i)) < 2$, we add the pair $(s_i, c_i')$ to $M'$. We claim that $M'$ is first-coalition-stable. Assume $M'$ is not, which implies the existence of a first-blocking coalition of $M'$ in $J$ involving one student $s_i$. Let us go through each possibility:

\begin{itemize}
    \item Student $s_i$ has only one credit in $M'$ (which must be from course $c_i'$, according to the construction of $M'$), and forms a first-blocking coalition with a course $c_j$, with one or two credits. In this case, hospital $h_i$ would have no residents in $M$, and resident $r_j$ is unmatched in $M$ or prefers $h_i$ to $M(r_j)$, 
    so $(r_j, h_i)$ blocks $M$ in $I$, a contradiction.
    \item Student $s_i$ has two credits in $M'$, one of which comes from $c_i'$ and the other from some $c_k$, and she forms a first-blocking coalition with course $c_j$ (and possibly with one other course also). As $c_i'$ is the least-preferred course on $s_i$'s preference list, it follows that $s_i$ prefers $c_j$ over $c_i'$.  Suppose firstly that $c_j$ has one credit. As $C_{M'}(s_i)=\{c_k,c_i'\}$, $(r_k,h_i)\in M$ where $r_k$ has size $1$. Moreover $h_i$ has room for resident $r_j$, also of size $1$. Resident $r_j$ is unmatched in $M$ or prefers $h_i$ over $M(r_j)$, and hence $(r_j,h_i)$ forms a blocking pair of $M$ in $I$, a contradiction. If $c_j$ has two credits, then $s_i$ prefers $c_j$ over $c_k$, which implies that in $I$ hospital $h_i$ prefers resident $r_j$ over its assigned resident $r_k$, and $r_j$ is unmatched in $M$ or prefers $h_i$ to $M(r_j)$, so $(r_j, h_i)$ blocks $M$ in $I$, a contradiction.
    \item Student $s_i$ has two credits in $M'$, which come from two courses with one credit, namely $c_p$ and $c_q$ (neither of which is equal to $c_i'$), where $c_p\succ_{s_i} c_q$, or one course with two credits, namely $c_p$. Suppose that $s_i$ forms a first-blocking coalition with either one course $c_j$, or with two courses, $c_j$ and $c_k$, such that $c_j \succ_{s_i} c_k$. Then, in $I$, hospital $h_i$ prefers resident $r_j$ over either resident $r_p$ or residents $r_p$ and $r_q$, and $r_j$ is unmatched in $M$ or prefers $h_i$ to $M(r_j)$, so $(r_j, h_i)$ blocks $M$ in $I$, a contradiction.
\end{itemize}

Similarly, assume that the {\sc ca} instance $J$ has a first-coalition-stable matching $M'$. Then we take $M'$, remove any pair containing a course of the form $c_i'$, and swap courses and students for residents and hospitals to obtain a matching $M$ in the {\sc hrs} instance $I$. This matching $M$ is also stable in $I$: if $(r_j, h_i)$ were to block $M$ in $I$, then either $(s_i, c_j)$ or $(s_i, \{c_j, c_i'\})$ would form a first-blocking coalition of $M$ in $J$, a contradiction.

Therefore, there exists a first-coalition-stable matching in $J$ if and only if there exists a stable matching in $I$. Thus, as {\sc hrs} is NP-complete, {\sc ca-fc-find} is NP-hard.
\end{proof}

\subsection{Master lists}
\label{subsec:findingMasterLists}
An important restriction of {\sc ca} arises when students are ordered in a master list, for example by Grade Point Average (GPA). Let {\sc ca-x-ml-find} be the problem of finding a stable matching in a {\sc ca-ml} instance $I$, for a given definition of stability {\sc x}, or determining that none exists. When students are ranked according to a master list, we can use a modification of the Serial Dictatorship mechanism to find a stable matching, for any definition of stability. 
The algorithm starts with an empty matching $M$ and, on round $i$, considers the $i$-th student in the master list, say $s_i$. Then, it considers every course on $s_i$'s preference list, from most to least preferred. For each such $c_j \in P(s_i)$, it checks whether $c_j$ has some capacity left, whether adding $(s_i, c_j)$ to $M$ does not violate the credit limit of $s_i$, and, if there are downward-feasible constraints, whether adding $(s_i, c_j)$ to $M$ does not violate any of them. If it is possible to add $(s_i, c_j)$ to $M$, it does so, and proceeds to check the next course in $P(s_i)$. This algorithm is described in pseudocode form in Algorithm \ref{alg:findMaster}. This is clearly an $O(L)$ time algorithm, where $L$ is the total number of acceptable pairs. We now show that Algorithm \ref{alg:findMaster} always finds a pair-stable (feasible) matching.

\begin{algorithm}[h]
\small
\caption{Algorithm to find an {\sc x}-stable (feasible) matching in an instance of {\sc ca-x-ml-find}}
\label{alg:findMaster}
\begin{algorithmic}[1]
\STATE $M := \emptyset$;
\FOR{$s_i$ in $\succ_{MLS}$ in decreasing order of preference}
   \STATE $p := 1$; ~~~~ \{$p$ points to first course in $s_i$'s list\}
   \WHILE{$O(C_M(s_i))<T(s_i)$ {\bf and} $p\leq |P(s_i)|$}
        \STATE $c_j := $ course at position $p$ of $\succ_{s_i}$;
        \IF{$O(C_M(s_i)) + O(c_j) \leq T(s_i)$ {\bf and} $|S_M(c_j)|<q^+(c_j)$ {\bf and} $C_M(s_i)\cup \{c_j\}$ is feasible}
            \STATE $M := M \cup \{(s_i,c_j)\}$;
       \ENDIF
       \STATE $p${\tt ++};
    \ENDWHILE;
\ENDFOR
\STATE{\bf return} $M$;\\
\end{algorithmic}
\end{algorithm}

\begin{proposition}
In any {\sc ca-ml} instance $I$, a pair-stable (feasible) matching always exists, and one can use Algorithm \ref{alg:findMaster} to find such a matching in $O(L)$ time, where $L$ is the total number of acceptable pairs. 
\label{prop:findingMasterList}
\end{proposition}
\begin{proof}
Aiming for a contradiction, assume that after executing Algorithm \ref{alg:findMaster}, student $s_i$ and course $c_j$ form a blocking pair. Then this means that either $c_j$ has some capacity left for $s_i$ or $c_j$ prefers $s_i$ over some $s_k$ that is assigned to $c_j$. If the second case holds, $s_i\succ_{MLS} s_k$, so $s_i$ got to choose courses before $s_k$, meaning that in the round of $s_i$, $s_k$ was not assigned to $c_j$ and $c_j$ had an open place that $s_i$ could have taken. If $s_i$ did not choose $c_j$ and instead chose other courses, then $s_i$ would prefer her assigned courses over $c_j$, which contradicts the assumption that $s_i$ and $c_j$ form a blocking pair. The reasoning is similar if the first case holds instead. Therefore, the algorithm above always produces a pair-stable matching.
\end{proof}

By Proposition \ref{prop:StabilityDefinitions}, a pair-stable matching is first-coalition-stable, coalition-stable, and pair-size-stable. Therefore, the matching found by Algorithm \ref{alg:findMaster} is first-coalition-stable, coalition-stable, and pair-size-stable, and thus {\sc ca-x-ml-find} is solvable in polynomial time for any definition of stability x.

Moreover, the pair-stable (feasible) matching found by Algorithm \ref{alg:findMaster} is the unique pair-stable (feasible) matching. Assume otherwise, so that there are two different pair-stable (feasible) matchings, namely $M$ and $M'$, and let $s_i$ be the highest-ranked student such that $C_M(s_i) \neq C_{M'}(s_i)$. That is, she is assigned a different set of courses in each matching. Let $c_j$ be the most-preferred course of $s_i$ in $C_M(s_i) \oplus C_{M'}(s_i)$, the symmetric difference of $C_M(s_i)$ and $C_{M'}(s_i)$. Without loss of generality, assume that $c_j \in C_M(s_i)$. Then $(s_i, c_j)$ blocks $M'$, because $s_i$ prefers to drop all courses in $C_{M'}(s_i) \setminus C_M(s_i)$ and be allocated to $c_j$ instead; 
this is the case because $O(C_M(s_i))\leq T(s_i)$ and $c_j\in C_M(s_i)\setminus C_{M'}(s_i)$ implies that $O(C_M(s_i)\cap C_{M'}(s_i))+O(c_j)\leq T(s_i)$; also the set $(C_M(s_i) \cap C_{M'}(s_i)) \cup \{(s_i, c_j)\}$ is feasible, as it is a subset of the feasible set $C_M(s_i)$; finally, because $s_i$ is the highest-ranked student that has different courses in $M$ and $M'$, $c_j$ prefers $s_i$ over the students that it has in $M'$ but not in $M$. Hence, we obtain a contradiction to the pair stability of $M'$, which implies that there is a unique pair-stable matching.

\begin{observation}
Every instance of {\sc ca-ml} has a unique pair-stable matching.
\label{obs:masterListUniquePairStable}
\end{observation}

Under other definitions of stability, an instance of {\sc ca-ml} might have more than one stable matching. For example, consider an instance with two courses, $c_1$ and $c_2$, such that $O(c_1) = 1$ and $O(c_2) = 2$, and one student, $s_1$, with $c_1 \succ_{s_1} c_2$ and $T(s_1) = 2$. Then both $\{(s_1, c_1)\}$ and $\{(s_1, c_2)\}$ are first-coalition-stable matchings, which implies that they are coalition-stable and pair-size-stable, too. $\{(s_1, c_2)\}$ is not a pair-stable matching though, because $s_1$ prefers $c_1$ and under pair stability she can end up with fewer credits when participating in a blocking pair.

\section{Finding a maximum size stable matching}
\label{sec:maximum}
\subsection{No lower quotas}
\label{subsec:maximumNoLowerQuotas}
As we saw in Section \ref{subsec:findingMasterLists}, every instance of {\sc ca-ml} has at least one stable matching, for any definition of stability. However, there may be more than one possible stable matching, and indeed they can have different sizes, in which case a natural goal is to find a stable matching that is as large as possible. We refer to this problem as {\sc ca-x-ml-max}, where {\sc x} is a definition of stability. {\sc ca-p-ml-max} is trivially solvable in polynomial time, by Observation \ref{obs:masterListUniquePairStable}.

However, under the other three definitions of stability, we will show in this section that {\sc ca-x-ml-max} is NP-hard. We will begin by showing that, under coalition stability, {\sc ca-c-ml-max} is NP-hard even if courses and students are ordered in two master lists, every course has one or two credits and upper quota $1$, and every student can take up to two credits.

In order to prove this result, we will transform from the {\sc stable marriage with ties and incomplete preferences} problem ({\sc smti}) \cite{manlove_2002_hard}. {\sc com-smti}, the problem of determining whether a complete weakly stable matching in an {\sc smti} instance exists, is NP-complete even if there are master lists on both sides, the ties occur in one master list only, and are of length $2$ \citep[][Thm.\ 3.2]{irving_stable_2008}.

Theorem \ref{thm:coalitionStabilityMaxSize}, which shows that {\sc ca-c-ml-max} is NP-hard, involves a key idea that is illustrated in Figure \ref{fig:exampleReductionCoalitionStability}. In this instance, $c_2$ is a $2$-credit course and $c_1$ and $c_3$ are $1$-credit courses. Thus, under coalition stability, $s_1$ is indifferent between being matched with $c_1$ and $c_3$, and with $c_2$. Hence, under coalition stability, we can speak of a tie between $c_2$ on the one hand, and $c_1$ and $c_3$ on the other, even though preference lists are all strict. This idea gives the connection between {\sc com-smti} and {\sc ca-c-ml-max}.

\begin{figure}[h]
    \centering
    \begin{align*}
        s_1&: c_1 \succ_{s_1} c_2 \succ_{s_1} c_3
    \end{align*}
    \caption{Preference list of $s_1$, with $T(s_1) = 2$. $O(c_1) = O(c_3) = 1$, and $O(c_2) = 2$.}
    \label{fig:exampleReductionCoalitionStability}
\end{figure}

\begin{theorem}
It is NP-hard to determine whether a {\sc ca-ml} instance has a student-complete (or course-complete) coalition-stable matching, even if both students and courses are ordered according to master lists, every student takes up to two credits, and every course has one or two credits and upper quota $1$. Therefore, {\sc ca-c-ml-max} is NP-hard, even if there are no downward-feasible constraints.
\label{thm:coalitionStabilityMaxSize}
\end{theorem}
\begin{proof}
We transform from the restriction of {\sc com-smti} in which both men women are ordered in two master lists, ties only occur in the women's master list (so on men's preference lists only), and each tie is of length $2$. Under these conditions, {\sc com-smti} is NP-complete \cite{irving_stable_2008}. Let $I$ be an instance of this restricted version of {\sc com-smti}. We construct a {\sc ca-ml} instance $J$ from $I$.

Every woman $w_j$ becomes a course $c_j$, with the same preference list, and with upper quota $1$. If $w_j$ is not involved in a tie in the women's master list then $c_j$ is a $2$-credit course. If women $w_k, w_l$ are tied in the women's master list, then one of them (say $w_k$) becomes a $1$-credit course ($c_k$), while the other ($w_l$) becomes a $2$-credit course ($c_l$). For every $1$-credit course $c_k$ we create another $1$-credit course $c_k'$. This course initially has the same preference list as $c_k$, 
and it also has upper quota $1$. The master list of courses is the same as the master list of women, except that if women $w_k, w_l$ are tied in the women's master list, then the master list of courses ($\succ_{MLC}$) includes $c_k \succ_{MLC} c_l \succ_{MLC} c_k'$, where $O(c_k) = O(c_k') = 1$ and $O(c_l) = 2$.

Every man $m_i$ becomes a student $s_i$, who can take up to two credits. On the preference list of every course in $J$, replace any man $m_i$ by student $s_i$. The preferences of $s_i$ are inherited from $m_i$, such that every woman $w_j$ not involved in a tie in the women's master list becomes $c_j$ in $s_i$'s preference list. If some $w_j$ is in a tie with $w_k$ in the women's master list, but only $w_j$ is in the preference list of $m_i$, then we replace $w_j$ by $c_j$, if $c_j$ is a $2$-credit course, or by $c_j$ first and then $c_j'$, if $c_j$ is a $1$-credit course. If both $w_j$ and $w_k$ are tied and in $m_i$'s preference list, then we substitute them both by $c_j \succ_{s_i} c_k \succ_{s_i} c_j'$. By Figure \ref{fig:exampleReductionCoalitionStability}, we know that $s_i$ is indifferent between being matched with $\{c_j, c_j'\}$ and with $\{c_k\}$, under coalition stability. From the men's master list we derive a students' master list ($\succ_{MLS}$). We now show that course-complete matchings are student-complete, and vice versa.

\begin{lemma}
Let $M'$ be a matching in $J$. Then $M'$ is course-complete if and only if it is student-complete.
\label{lem:ThmCoalitionStabilityComplete}
\end{lemma}
\begin{proof}\renewcommand{\qedsymbol}{\ensuremath{\blacksquare}}
Let $n$ be the number of men (and women) in the {\sc com-smti} instance $I$. Then, in the {\sc ca-ml} instance $J$, for each woman $w_j$ we have either one $2$-credit course with upper quota $1$, or two $1$-credit courses with upper quota $1$. Thus, a course-complete matching has a total of $2n$ credits. Similarly, for each man we have a student $s_i$ with $T(s_i) = 2$. Thus, a student-complete matching has $2n$ credits. Therefore, a course-complete matching has the same number of credits as a student-complete matching.
\end{proof}

We next show that in a coalition-stable matching $M'$ in $J$, a student $s_i$ is matched with either a $2$-credit course or two $1$-credit courses that come from the same woman in $I$.

\begin{lemma}
Let $M'$ be a coalition-stable matching in $J$. Then any student $s_i$ is matched in $M'$ with either a $2$-credit course $c_j$ or two $1$-credit courses $\{c_k, c_k'\}$, such that both courses come from the same woman $w_k$ in $I$.
\label{lem:ThmCoalitionStabilityCourses}
\end{lemma}
\begin{proof}\renewcommand{\qedsymbol}{\ensuremath{\blacksquare}}
Assume that this is not the case, and that there exists some $s_i$ that is matched in $M'$ with a $1$-credit course $c_j$ but not with $c_j'$ (the other case, where $s_i$ is matched to $c_j'$ instead, can be proved analogously). Furthermore, assume that $s_i$ is the highest placed student in $\succ_{MLS}$ in this situation, and without loss of generality assume that $c_j$ is ranked higher in $\succ_{MLC}$ than the other course matched to $s_i$, if there is one. Then, because $s_i$ is the highest ranked student in $\succ_{MLS}$ matched to one but not both of the courses that come from the same woman, $c_j'$ is unmatched or prefers $s_i$ over its only student. Similarly, if $s_i$ is matched to some $1$-credit course $c_l$, then she prefers $c_j'$ over $c_l$. Therefore, $(s_i, \{c_j'\})$ is a blocking coalition, a contradiction to the stability of $M'$. Hence, if a student $s_i$ is matched with a $1$-credit course $c_j$ in $M'$, she is also matched with $c_j'$ in $M'$.
\end{proof}

Now we show that there exists a complete weakly stable matching in $I$ if and only if there exists a student-complete (or course-complete by Lemma \ref{lem:ThmCoalitionStabilityComplete}) coalition-stable matching in $J$. First, assume that there is a complete weakly stable matching $M$ in $I$. Then we build a complete matching $M'$ in $J$ as follows. For each $(m_i, w_j) \in M$ such that $c_j$ is a $2$-credit course, add $(s_i, c_j)$ to $M'$. For each $(m_i, w_j) \in M$ such that $c_j$ is a $1$-credit course, add $(s_i, c_j)$ and $(s_i, c_j')$ to $M'$. Clearly the matching is both student-complete and course-complete, as every course is matched, and every student takes two credits. We now show that $M'$ is coalition-stable. Assume that it is not, and that there exists a blocking coalition $(s_p, \{c_q\})$, with $O(c_q) = 2$ (or $(s_p, \{c_q, c_q'\})$ where $O(c_q) = 1$). Thus, $s_p$ prefers $c_q$ (and $c_q'$, if it exists) over every course in $C_{M'}(s_p)$. This implies that $c_q$ (and $c_q'$, if it exists) is placed higher in $\succ_{MLC}$ than every course in $C_{M'}(s_p)$. Similarly, $c_q$ (and $c_q'$, if it exists) prefers $s_p$ over its student $s_r$, meaning that $s_p$ is placed higher in $\succ_{MLS}$ than $s_r$. However, this means than in $I$, $w_q$ prefers $m_p$ over $m_r$, and that $m_p$ prefers $w_q$ over his partner in $M$, meaning that $(m_p, w_q)$ blocks $M$, a contradiction to the weak stability of $M$ in $I$. Thus, $M'$ is a coalition-stable matching in $J$.

Now, assume that $J$ has a student-complete (course-complete) coalition-stable matching $M'$. Then, by Lemma \ref{lem:ThmCoalitionStabilityCourses}, every student $s_i$ is matched with either a $2$-credit course $c_j$, or two $1$-credit courses, $c_k$ and $c_k'$. We build a complete matching $M$ in $I$ as follows. If $(s_i, c_j) \in M'$, add $(m_i, w_j)$ to $M$. Clearly $M$ is complete, because every man and every woman is matched. We now show that $M$ is weakly stable. Assume that it is not, so that there exists a blocking pair $(m_p, w_q)$. Thus, $m_p$ prefers $w_q$ over his partner in $M$, and $w_q$ prefers $m_p$ over her partner in $M$, which means that $w_q$ and $m_p$ are ranked higher in their respective master lists than their partners. This implies that $c_q$ (or $c_q$ and $c_q'$) and $s_p$ are ranked higher in their respective master lists than the courses in $C_{M'}(s_p)$ and the only student in $S_{M'}(c_q)$, respectively. Hence, $s_p$ prefers $c_q$ (and $c_q'$, if it exists) over every course in $C_{M'}(s_p)$. As $c_q$ (and $c_q'$) also prefers $s_p$ over its only assigned student, $(s_p, \{c_q\})$ (or $(s_p, \{c_q, c_q'\})$) is a blocking coalition of $M'$, which contradicts the coalition stability of $M'$ in $J$. Thus, $M$ is a weakly stable matching in $I$.

Therefore, there exists a complete weakly stable matching in $I$ if and only if there exists a student-complete (course-complete) coalition-stable matching in $J$. Hence, determining whether a {\sc ca-ml} instance has a student-complete (course-complete) coalition-stable matching is NP-hard. Thus, {\sc ca-c-ml-max} is NP-hard.
\end{proof}

Likewise, determining whether a {\sc ca-ml} instance has a course-complete first-coalition-stable matching is NP-hard, as well as determining whether it has a course-complete pair-size-stable matching, as we will now show. Hence, both {\sc ca-fc-ml-max} and {\sc ca-ps-ml-max} are NP-hard. 

\begin{theorem}
It is NP-hard to determine whether there exists a course-complete first-coalition-stable matching, even if both students and courses are ordered according to master lists, and every course has upper quota $1$. The same result also holds for pair-size stability. Hence, both {\sc ca-fc-ml-max} and {\sc ca-ps-ml-max} are NP-hard, even if there are no downward-feasible constraints.
\label{thm:firstCoalitionStabilityMaxSizeNcredits}
\end{theorem}
\begin{proof}
We again transform from the NP-complete restriction of {\sc com-smti} in which both men and women are ordered in two master lists, ties only occur in the women's master list, and each tie is of length $2$ \cite{irving_stable_2008}. Let $I$ be an instance of this restricted version {\sc com-smti}.  We construct a {\sc ca-ml} instance $J$ from $I$.

Every woman $w_j$ becomes a course $c_j$ with upper quota $1$. 
The master list of courses $\succ_{MLC}$ in $J$ is obtained from the master list of women in $I$ by replacing every woman $w_j$ by course $c_j$. We will ensure that any two courses have a different number of credits, as follows. Assume that the courses in $\succ_{MLC}$ are ordered $c_1,c_2,\dots,c_n$ such that 
$c_1\succeq_{MLC} c_2\succeq_{MLC} \dots c_{n-1}\succeq_{MLC} c_n$,
where $n$ is the number of courses, and $c_r\succeq_{MLC} c_s$ denotes that $c_r$ comes before $c_s$ on $\succ_{MLC}$ or the two courses are tied in $\succ_{MLC}$.  For each $j$ ($1\leq j\leq n$), let $c_j$ have number of credits $O(c_j)= 1 + (n-j+1)\varepsilon$, where $\varepsilon$ is an arbitrarily small positive number (e.g., $\varepsilon=\frac{1}{n}$).
Next, if two courses $c_j$ and $c_{j+1}$ are tied in $\succ_{MLC}$, swap $O(c_j)$ and $O(c_{j+1})$.  This ensures that if $c_r\succ_{MLC} c_s$ then $O(c_r)>O(c_s)$, whilst if $c_r$ and $c_s$ are tied in $\succ_{MLC}$ then $O(c_r)<O(c_s)$.  Hence a student matched with $c_s$ cannot form a first-blocking coalition with $c_r$, and vice versa.

Every man $m_i$ becomes a student $s_i$, with $T(s_i) = 1 + n\cdot\varepsilon$. The master list of students $\succ_{MLS}$ in $J$ is obtained from the master list of men in $I$ by replacing every man $m_i$ by student $s_i$.  
It is easy to see that no student has enough available credits for more than one course.

First, assume that there is a complete weakly stable matching $M$ in $I$. Then we build a course-complete matching $M'$ in $J$ as follows. If $(m_i, w_j) \in M$, add $(s_i, c_j)$ to $M'$. The matching is course-complete, as every course is matched with one student. We now show that $M'$ is first-coalition-stable in $J$. Assume that it is not, and that there exists a first-blocking coalition $(s_p, \{c_q\})$. Thus, $s_p$ prefers $c_q$ over the course in $C_{M'}(s_p)$, and $O(c_q) \geq O(C_{M'}(s_p))$. This implies that $c_q$ is placed higher in $\succ_{MLC}$ than the course in $C_{M'}(s_p)$. Similarly, $c_q$ prefers $s_p$ over its only matched student $s_r$, meaning that $s_p$ is placed higher in $\succ_{MLS}$ than $s_r$. However, this means than in $I$, $w_q$ prefers $m_p$ over $m_r$, and that $m_p$ prefers $w_q$ over his partner, meaning that $(m_p, w_q)$ blocks $M$, a contradiction to the weak stability of $M$ in $I$. Thus, $M'$ is a first-coalition-stable matching in $J$.

Now, assume that $J$ has a course-complete first-coalition-stable matching $M'$. Thus, every student $s_i$ is matched with one course, because no student has enough available credits to be matched with two courses. We build a complete matching $M$ in $I$ as follows. For each $(s_i, c_j) \in M'$, add $(m_i, w_j)$ to $M$. The matching is complete, as every man is matched with one woman. We now show that $M$ is weakly stable in $I$, too. Assume that it is not, so that there exists a blocking pair $(m_p, w_q)$. Thus, $m_p$ prefers $w_q$ over his partner, and $w_q$ prefers $m_p$ over her partner, which means that $w_q$ and $m_p$ are ranked higher in their respective master lists than their respective partners. This implies that $c_q$ and $s_p$ are ranked higher in their respective master lists than the course in $C_{M'}(s_p)$ and the student in $S_{M'}(c_q)$, respectively. Hence, $s_p$ prefers $c_q$ over the course in $C_{M'}(s_p)$, and $O(c_q) > O(C_{M'}(s_p))$. As $c_q$ also prefers $s_p$ over its only assigned student, $(s_p, \{c_q\})$ is a first-blocking coalition, which contradicts the first-coalition stability of $M'$ in $J$. Thus, $M$ is a weakly stable matching in $I$.

Therefore, there exists a complete weakly stable matching in $I$ if and only if there exists a course-complete first-coalition-stable matching in $J$. Hence, determining whether a {\sc ca-ml} instance has a course-complete first-coalition-stable matching is NP-hard. This same proof can be used to show that determining the existence of a course-complete pair-size-stable matching is NP-hard.
\end{proof}

In fact, the problem of finding a maximum size first-coalition-stable matching in {\sc ca-ml} remains NP-hard even if every course has one or two credits. In order to prove this result, we will reduce from {\sc min-mm}. In this problem, we have a graph $G$ and a positive integer $K$, and the goal is to determine whether there exists a maximal matching of size at most $K$. This problem is NP-complete, even if $G$ is a subdivision graph \cite{horton_1993_minimum}, which implies that it is bipartite and all vertices in one of the sides have degree $2$.

\begin{theorem}
It is NP-hard to determine whether an instance of {\sc ca-ml} has a course-complete first-coalition-stable matching, even if 
every student takes up to two credits, every course has one or two credits, and each student has up to $5$ courses in her preference list.
\label{thm:firstCoalitionStabilityMaxSize}
\end{theorem}
\begin{proof}
We reduce from {\sc min-mm}, which is NP-complete, even if $G$ is a subdivision graph \cite{horton_1993_minimum}. Let $G = (U, W, E)$ be a graph and let $K$ be an integer representing an instance of {\sc min-mm}, where $U$ and $W$ are the two sets of vertices in the bipartition, $|U| = n_1$ and $|W| = n_2$, and $K$ is the target size of a maximal matching. Without loss of generality assume that every vertex in $W$ has degree $2$. We build an instance $I$ of {\sc ca-ml}, with students $S$ and courses $C$, as follows.

For every $w_i \in W$, we create the following students: $s_i^1$, $s_i^2$, and $s_i^3$, with $T(s_i^1) = T(s_i^2) = T(s_i^3) = 2$, as well as $s_i^4$, with $T(s_i^4) = 1$. We also create courses $d_i^1$, $d_i^2$, and $d_i^3$, with $O(d_i^1) = O(d_i^2) = 2$, and $O(d_i^3) = 1$. Create the following subsets: $S_i=\{s_i^1,s_i^2,s_i^3,s_i^4\}$ ($1\leq i\leq n_2$), $D_i=\{d_i^r : 1\leq i\leq n_2\wedge 1\leq r\leq 3\}$, $S^r = \{s_i^r : 1\leq i\leq n_2\}$ ($1\leq r\leq 4$), and $D^r = \{d_i^r : 1\leq i\leq n_2\}$ ($1\leq r\leq 3$).

For every $u_j \in U$, we create a course $c_j$, with $O(c_j) = 1$, and a student $p_j$, with $T(p_j) = 1$. Additionally, we create a course $e$, $O(e) = 2$. The upper quota of every course is $1$, except for $e$, which has upper quota $n_2-K$. Let $C'=\{c_j : 1\leq j\leq n_1\}$ and $P=\{p_j : 1\leq j\leq n_1\}$.

The master list of students $(\succ_{MLS})$ is
\begin{multline*}
s_1^1 \succ_{MLS} s_1^2 \succ_{MLS} s_2^1 \succ_{MLS} s_2^2 \succ_{MLS} \dots \succ_{MLS} s_{n_2}^1 \succ_{MLS} s_{n_2}^2 \succ_{MLS} \\ s_1^3 \succ_{MLS} s_2^3 \succ_{MLS} \dots \succ_{MLS} s_{n_2}^3 \succ_{MLS} p_1 \succ_{MLS} p_2 \succ_{MLS} \dots \succ_{MLS} p_{n_1} \succ_{MLS} \\ s_1^4 \succ_{MLS} s_2^4 \succ_{MLS} \dots \succ_{MLS} s_{n_2}^4
\end{multline*}

That is, the first students belong to subsets $S^1$ and $S^2$, such that every $s_i^1$ is followed by $s_i^2$. Then we place students from $S^3$, then students from $P$, and finally students from $S^4$, in all cases in increasing subscript order.



Let $w_i \in W$ be a vertex from $G$, adjacent to $u_j$ and $u_k$ (with $j < k$). We build the preference lists of $s_i^1$, $s_i^2$, $s_i^3$ and $s_i^4$ as follows.
\begin{align*}
s_i^1&: c_j \succ_{s_i^1} d_i^1 \\
s_i^2&: c_k \succ_{s_i^2} d_i^1 \succ_{s_i^2} d_i^2 \\
s_i^3&: d_i^2 \succ_{s_i^3} d_i^3 \succ_{s_i^3} c_j \succ_{s_i^3} c_k \succ_{s_i^3} e \\
s_i^4&: d_i^3
\end{align*}

If in the {\sc min-mm} instance $w_i$ is matched with $u_j$, then $s_i^1$ is to be matched with $c_j$, $s_i^2$ with $d_i^1$, $s_i^3$ with $d_i^2$ and $s_i^4$ with $d_i^3$. Similarly if $w_i$ is matched with $u_k$, then $s_i^1$ is to be matched with $d_i^1$, $s_i^2$ with $c_k$, $s_i^3$ with $d_i^2$ and $s_i^4$ with $d_i^3$. If, however, $w_i$ is unmatched, $s_i^1$ is to be matched with $d_i^1$, $s_i^2$ with $d_i^2$, $s_i^3$ with $e$ and $s_i^4$ with $d_i^3$. Student $s_i^3$ contains courses $c_j$, $c_k$, and $d_i^3$ in her preference list to make sure that a course-complete first-coalition-stable matching in $I$ corresponds to a maximal matching in $G$, so that if $s_i^3$ is matched with $e$, some students from $S^1$ or $S^2$ are matched with $c_j$ and $c_k$. For each student $p_j\in P$, her preference list contains course $c_j$ only:
\begin{align*}
p_j: c_j
\end{align*}

The preference list of every $c_j \in C'$ depends on the vertices adjacent to $u_j$ in $G$. For every $w_i$ adjacent to $u_j$, $c_j$ includes in its preference list either $s_i^1$ or $s_i^2$, as well as $s_i^3$. Additionally, it includes student $p_j$. The order of all students is derived from the master list. If, for example, $u_j$ is adjacent to $w_3$, $w_4$ and $w_5$, the preference list of $c_j$ is given below. Here the superscript $l_i$ indicates a number that can be either $1$ or $2$. If $u_j$ is matched in $G$ then $c_j$ is matched with some student from $S^1$, $S^2$ or $S^3$, otherwise it is matched with $p_j$.
\begin{align*}
c_j: s_3^{l_3} \succ_{c_j} s_4^{l_4} \succ_{c_j} s_5^{l_5} \succ_{c_j} s_3^3 \succ_{c_j} s_4^3 \succ_{c_j} s_5^3 \succ_{c_j} p_j
\end{align*}

The preference lists of every course $d_i^1 \in D^1$, $d_i^2 \in D^2$, and $d_i^3 \in D^3$ are:
\begin{align*}
d_i^1&: s_i^1 \succ_{d_i^1} s_i^2 \\
d_i^2&: s_i^2 \succ_{d_i^2} s_i^3 \\
d_i^3&: s_i^3 \succ_{d_i^3} s_i^4
\end{align*}

Finally, the preference list of course $e$ contains every student from $S^3$ in increasing subscript order:
\begin{align*}
e: s_1^3 \succ_e s_2^3 \succ_e ... \succ_e s_{n_2}^3
\end{align*}
The students from $S^3$ matched to $e$ correspond to vertices in $G$ that are unmatched.

An example of this transformation is shown in Figure \ref{fig:ExampleTheoremFirstCoalition}.

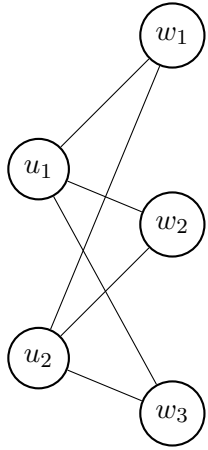
\begin{figure}[ht]
    \centering
    \begin{subfigure}[b]{0.35\linewidth}
        \centering
        \begin{tikzpicture}[main/.style = {node distance={25mm}, thick, draw, circle}]
            \node[main] (1) {$u_1$};
            \node[main] (2) [below of=1] {$u_2$};
            \node[main] (3) [above right of=1] {$w_1$};
            \node[main] (4) [below of=3] {$w_2$};
            \node[main] (5) [below of=4] {$w_3$};
            \draw (1) -- (3);
            \draw (1) -- (4);
            \draw (1) -- (5);
            \draw (2) -- (3);
            \draw (2) -- (4);
            \draw (2) -- (5);
        \end{tikzpicture}
    \end{subfigure}
    \hspace{0.05\linewidth}
    \begin{subfigure}[b]{0.55\linewidth}
        \centering
        \begin{align*}
        s_1^1&: c_1 \succ_{s_1^1} d_1^1 \\
        s_1^2&: c_2 \succ_{s_1^2} d_1^1 \succ_{s_1^2} d_1^2 \\
        s_1^3&: d_1^2 \succ_{s_1^3} d_1^3 \succ_{s_1^3} c_1 \succ_{s_1^3} c_2 \succ_{s_1^3} e \\
        s_1^4&: d_1^3 \\ \\
        c_1&: s_1^1 \succ_{c_1} s_2^1 \succ_{c_1} s_3^1 \succ_{c_1} s_1^3 \succ_{c_1} s_2^3 \succ_{c_1} s_3^3 \succ_{c_1} p_1 \\
        e&: s_1^3 \succ_{e} s_2^3 \succ_{e} s_3^3
        \end{align*}
    \end{subfigure}
    \caption{Example of a {\sc min-mm} instance, with $K=2$, and some of the preference lists in the transformed instance of {\sc ca-ml}}
    \label{fig:ExampleTheoremFirstCoalition}
\end{figure}

We first show that a maximal matching $M$ of size at most $K$ in $G$ implies a course-complete first-coalition-stable matching $M'$ in $I$. For every pair $(u_j, w_i) \in M$, we add to $M'$ either the pair $(s_i^1, c_j)$ or the pair $(s_i^2, c_j)$, depending on whether $s_i^1$ or $s_i^2$ has $c_j$ in her preference list. If $(s_i^1, c_j) \in M'$, then we add $(s_i^2, d_i^1)$ to $M'$, otherwise we add $(s_i^1, d_i^1)$. We also add $(s_i^3, d_i^2)$ and $(s_i^4, d_i^3)$ to $M'$. If $w_i$ has no partner in $M$, we add the following pairs to $M'$: $(s_i^1, d_i^1)$, $(s_i^2, d_i^2)$, $(s_i^3, e)$, and $(s_i^4, d_i^3)$. Finally, if vertex $u_j$ is not matched in $M$, add $(p_j, c_j)$ to $M'$. The matching $M'$ is course-complete, as every course in $C'$ or $D$ is matched with one student, and $e$ is matched with $n_2-K$ students. It is also first-coalition-stable, as we show below.

\begin{lemma}
$M'$ is a first-coalition-stable matching in $I$.
\end{lemma}
\begin{proof}\renewcommand{\qedsymbol}{\ensuremath{\blacksquare}}
We prove the first-coalition stability of $M'$ by showing that no student can belong to a first-blocking coalition. Every $s_i^1 \in S^1$ is matched with either some $c_j \in C'$, which is her most-preferred course, or with $d_i^1$. As $O(d_i^1) > O(c_j)$, $(s_i^1, \{c_j\})$ cannot form a first-blocking coalition. Similarly, every $s_i^2 \in S^2$ is matched with either some $c_k \in C'$, her most-preferred course, with $d_i^1$, or with $d_i^2$. If $s_i^2$ is matched with $c_k$ or with $d_i^1$, the same reasoning applies as for $s_i^1$, and if she is matched with $d_i^2$, then $(s_i^1, d_i^1) \in M'$, and $s_i^1 \succ_{d_i^1} s_i^2$, so there is no blocking coalition containing $s_i^2$. Every $s_i^4 \in S^4$ is matched with the only course in her preference list, $d_i^3$.

We now consider $s_i^3 \in S^3$. This student is matched with either $d_i^2$ or with $e$. In the first case, she gets her most-preferred course. In the second case, $(s_i^2, d_i^2) \in M$, $s_i^2 \succ_{d_i^2} s_i^3$, and, because $M$ is maximal in $G$, $u_j$ and $u_k$ are matched in $M$ with some $w_p$ and $w_q$ respectively. Thus, there exist $s_p^{l_p},s_q^{l_q}\in S^1\cup S^2$ such that $(s_p^{l_p},c_j)\in M'$ and $(s_q^{l_q},c_k)\in M'$. Thus, $s_i^3$ could form a first-blocking coalition with $d_i^3$ only. However, $O(d_i^3) < O(e)$, so $(s_i^3, \{d_i^3\})$ is not a first-blocking coalition. Finally, every $p_j \in P$ is either matched with $c_j$, the only course in her preference list, or she has no partner. In this second case, $c_j$ is matched with some $s_i^1$ or $s_i^2$, so $c_j$ prefers its current student over $p_j$, and thus $p_j$ cannot be part of a first-blocking coalition.
\end{proof}

Now, assume that there exists a course-complete first-coalition-stable matching $M'$ in $I$. We now show how to obtain a maximal matching $M$ of size at most $K$ in $G$. We begin by establishing two lemmas that give useful properties of $M'$.
\begin{lemma}
There are $n_2-K$ subsets $S_i$ for which $M'$ includes no pair between a member of $S_i$ and a member of $C'$.
\label{lem:ThmFirstCoalitionStabilityMaxSizeProperty1}
\end{lemma}
\begin{proof}\renewcommand{\qedsymbol}{\ensuremath{\blacksquare}}
As $M'$ is course-complete, $|S_{M'}(e)| = q^+(e) = n_2-K$, so there are $n_2-K$ students $s_i^3 \in S^3$ matched to $e$. For each pair $(s_i^3, e) \in M'$, $s_i^3$ would prefer to be matched with $d_i^2$. Thus, because $M'$ is first-coalition-stable, we must have that $(s_i^2, d_i^2) \in M'$, which in turn implies that $(s_i^1, d_i^1) \in M'$. Likewise, $s_i^3$ prefers to be matched with $d_i^3$ and with either $c_j$ or with $c_k$. Hence, because $M'$ is first-coalition-stable, and since $d_i^3$ ranks $s_i^3$ first, both $c_j$ and $c_k$ must have better students than $s_i^3$, so students from $S^1$ or $S^2$ (or possibly $S^3$). We also infer that, 
due to $M'$ being course-complete, $(s_i^4, d_i^3) \in M'$. Thus, if $(s_i^3, e) \in M'$, no student from $S_i$ is matched with a course from $C'$. These $S_i$ subsets correspond to $n_2-K$ vertices $w_i \in G$ that are not matched in the maximal matching $M$.
\end{proof}


Let us now consider the $K$ remaining $S_i$ subsets. For some of these subsets, there is no $s_i^{l_i} \in S_i$ and $c_j \in C'$ such that $(s_i^{l_i}, c_j) \in M'$. In these cases, we have that $(s_i^1, d_i^1) \in M'$, $(s_i^2, d_i^2) \in M'$ and $(s_i^3, d_i^3) \in M'$, as $M'$ is first-coalition-stable and course-complete, so every $d_i^l \in D_i$ must have one student. In these cases, we have that both $c_j$ and $c_k$ are matched with students from $S^1$, $S^2$ or $S^3$, to prevent $s_i^3$ forming a first-blocking coalition with either $c_j$ or $c_k$. For the remaining subsets $S_i$, there exists exactly one pair of the form $(s_i^l, c_j)$ in $M'$.

\begin{lemma}
If $(s_i^{l_i}, c_j) \in M'$ for some $s_i^{l_i} \in S_i$, then $M'$ includes no other pair between a member of $S_i$ and a member of $C'$.
\label{lem:ThmFirstCoalitionStabilityMaxSizeProperty2}
\end{lemma}
\begin{proof}\renewcommand{\qedsymbol}{\ensuremath{\blacksquare}}
If $s_i^1$ is matched with $c_j$, then $(s_i^2, d_i^1)$ must be in $M'$, otherwise $d_i^1$ has no students and $M'$ is not course-complete. Thus, $(s_i^3, d_i^2) \in M'$ and $(s_i^4, d_i^3) \in M'$. If $s_i^2$ is matched with $c_k$, then, by the same reasoning, $(s_i^1, d_i^1) \in M'$, $(s_i^3, d_i^2) \in M'$ and $(s_i^4, d_i^3) \in M'$. Finally, if $(s_i^3, c_j) \in M'$ (or $(s_i^3, c_k) \in M'$), then, to preserve first-coalition stability, $(s_i^3, d_i^3) \in M'$, $(s_i^1, d_i^1) \in M'$, and $(s_i^2, d_i^2) \in M'$.
\end{proof}


Hence, for every $S_i$ $(1\leq i\leq n_2)$, Lemmas \ref{lem:ThmFirstCoalitionStabilityMaxSizeProperty1} and \ref{lem:ThmFirstCoalitionStabilityMaxSizeProperty2} imply that either one $s_i^{l_i} \in S_i$ is matched in $M'$ with some $c_j\in C'$ or no member of $S_i$ is matched in $M'$ with a member of $C'$. If some $s_i^{l_i}\in S_i$ is matched with some $c_j$ in $M'$, add $(u_j, w_i)$ to $M$ in $G$, which is a valid matching in $I$ because every vertex is matched with at most one, adjacent, vertex. As there are at least $n_2-K$ subsets where no member is matched with some $c_j$, we have that there are at most $K$ $S_i$ subsets where exactly one member is matched with one $c_j$. This means that the matching $M$ is of size at most $K$. We now show that $M$ is maximal.

\begin{lemma}
The matching $M$ in $G$, of size at most $K$, is maximal.
\end{lemma}
\begin{proof}\renewcommand{\qedsymbol}{\ensuremath{\blacksquare}}
Looking for a contradiction, assume that $M$ is not maximal, so that there exists unmatched $u_j, w_i$ that are adjacent to each other in $G$. This implies that in $M'$, $c_j$ is matched with $p_j$, and $s_i^3$ is matched with either $e$ or $d_i^3$. In the first case, $d_i^3$ is matched with $s_i^4$ in $M'$. However, $s_i^3$ prefers $\{d_i^3, c_j\}$ over $e$, $c_j$ prefers $s_i^3$ over $p_j$, and $d_i^3$ prefers $s_i^3$ over $s_i^4$. Thus, $(s_i^3, \{d_i^3, c_j\})$ is a first-blocking coalition of $M'$.

In the second case, $s_i^3$ still has capacity for $c_j$, as $O(c_j) = d_i^3) = 1$ and $T(s_i^3) = 2$, and $c_j$ prefers $s_i^3$ over $p_j$. Thus, $(s_i^3, \{c_j\})$ is a first-blocking coalition of $M'$. In either case we get a contradiction to the first-coalition stability of $M'$ in $I$. Therefore, $M$ is maximal in $G$.
\end{proof}

Hence, there exists a maximal matching of size at most $K$ in $G$ if and only if there exists a course-complete first-coalition-stable matching in $I$. As {\sc min-mm} is NP-complete, determining the existence of a course-complete first-coalition-stable matching in a {\sc ca-ml} instance is NP-hard.
\end{proof}

In fact, it is easy to see that the proof above applies to coalition stability as well. Hence, we obtain the following result.

\begin{corollary}
Determining whether a {\sc ca-ml} instance has a course-complete coalition-stable matching is NP-hard, even if every student takes up to two credits, every course has one or two credits, and each student's preference list is of length at most $5$.
\label{cor:coalitionStabilityMaxSize}
\end{corollary}

The main difference between this corollary and Theorem \ref{thm:coalitionStabilityMaxSize} is that in Theorem \ref{thm:coalitionStabilityMaxSize}, courses are also ordered in a master list and every course has upper quota $1$, but the length of students' preference lists is not bounded.

We now show that determining the existence of a student-complete (or course-complete) pair-size-stable matching is NP-hard. In order to prove this result, we will reduce from {\sc exact-mm}. In this problem, we have a graph $G$ and a positive integer $K$, and the goal is to determine whether there exists a maximal matching of size exactly $K$. This problem is NP-complete, even if $G$ is a subdivision graph \cite{manlove_2002_hard}\citep[][page 133]{manlove_algorithmics_2013}, which implies that it is bipartite and all vertices in one of the sides have degree $2$.

\begin{theorem}
It is NP-hard to determine whether an instance of {\sc ca-ml} has a student-complete (or course-complete) pair-size-stable matching, even if both students and courses are ordered according to master lists, every student takes up to two credits, and every course has one or two credits.
\label{thm:pairSizeStabilityMaxSize}
\end{theorem}
\begin{proof}
We reduce from {\sc exact-mm}, which is NP-complete, even if $G$ is bipartite and all vertices in one of the sides have degree $2$ \cite{manlove_2002_hard}\citep[][page 133]{manlove_algorithmics_2013}. Let $G = (U, W, E)$ be a graph and $K$ be an integer representing an instance of {\sc exact-mm}, where $U$ and $W$ are the two sets of vertices in the bipartition, $|U| = n_1$ and $|W| = n_2$, and $K$ is the target size of a maximal matching. Without loss of generality assume that every vertex in $W$ has degree $2$. We build an instance $I$ of {\sc ca-ml}, with students $S$ and courses $C$, as follows.

For every $w_i \in W$, we create the following students: $s_i^1$, $s_i^2$, and $s_i^3$, with $T(s_i^1) = T(s_i^2) = T(s_i^3) = 2$. We also create courses $d_i^1$ and $d_i^2$, with $O(d_i^1) = O(d_i^2) = 2$. Create the following subsets: $S_i=\{s_i^1,s_i^2,s_i^3\}$ ($1\leq i\leq n_2$), $D_i=\{d_i^1,d_i^2\}$ ($1\leq i\leq n_2$), $S^r = \{s_i^r : 1\leq i\leq n_2\}$ ($1\leq r\leq 3$), and $D^r = \{d_i^r : 1\leq i\leq n_2\}$ ($1\leq r\leq 2$).

For every $u_j \in U$, we create a course $c_j$, with $O(c_j) = 1$. Additionally, we create three courses $e_1$, $e_2$ and $f$, with $O(e_1) = O(e_2) = O(f) = 1$, as well as $n_1 - K$ students labelled $p_1$ to $p_{n_1 - K}$, such that, for every $p_k$, $T(p_k) = 1$. The upper quota of every course is $1$, except for $e_1$ and $e_2$, which both have upper quota $n_2-K$, and $f$, which has upper quota $K$. Let $C'=\{c_j : 1\leq j\leq n_1\}$ and let $P=\{p_k : 1\leq k\leq n_1-K\}$.

The master list of students ($\succ_{MLS})$ is
\begin{multline}
s_1^1 \succ_{MLS} s_1^2 \succ_{MLS} s_2^1 \succ_{MLS} s_2^2 \succ_{MLS} \dots \succ_{MLS} s_{n_2}^1 \succ_{MLS} s_{n_2}^2 \succ_{MLS} \\ s_1^3 \succ_{MLS} s_2^3 \succ_{MLS} \dots \succ_{MLS} s_{n_2}^3 \succ_{MLS} p_1 \succ_{MLS} p_2 \succ_{MLS} \dots \succ_{MLS} p_{n_1-K}
\end{multline}

That is, the first students belong to subsets $S^1$ and $S^2$, such that every $s_i^1$ is followed by $s_i^2$. Then we place students from $S^3$, and finally students from $P$, in all cases in increasing subscript order. The master list of courses ($\succ_{MLC}$) is
\begin{multline}
c_1 \succ_{MLC} c_2 \succ_{MLC} \dots \succ_{MLC} c_{n_1} \succ_{MLC} d_1^1 \succ_{MLC} d_1^2 \succ_{MLC} d_2^1 \succ_{MLC} d_2^2 \succ_{MLC} \dots \succ_{MLC} \\ d_{n_2}^1 \succ_{MLC} d_{n^2}^2 \succ_{MLC} e_1 \succ_{MLC} e_2 \succ_{MLC} f
\end{multline}

That is, the first courses belong to $C'$, then the next courses are from $D^1$ and $D^2$, such that for each $w_i \in W$, $d_i^1 \succ_{MLC} d_i^2$, and finally we have courses $e_1$, $e_2$ and $f$.

Let $w_i \in W$ be a vertex from $G$, adjacent to $u_j$ and $u_k$ (with $j < k$). We build the preference lists of $s_i^1$, $s_i^2$ and $s_i^3$ as follows.
\begin{align*}
s_i^1&: c_j \succ_{s_i^1} d_i^1 \succ_{s_i^1} f \\
s_i^2&: c_k \succ_{s_i^2} d_i^1 \succ_{s_i^2} d_i^2 \succ_{s_i^2} f \\
s_i^3&: c_j \succ_{s_i^3} c_k \succ_{s_i^3} d_i^2 \succ_{s_i^3} e_1 \succ_{s_i^3} e_2
\end{align*}

If in the {\sc exact-mm} instance $w_i$ is matched with $u_j$, then $s_i^1$ is to be matched with $c_j$ and $f$, $s_i^2$ with $d_i^1$ and $s_i^3$ with $d_i^2$. Similarly if $w_i$ is matched with $u_k$, then $s_i^1$ is to be matched with $d_i^1$, $s_i^2$ with $c_k$ and $f$, and $s_i^3$ with $d_i^2$. If, however, $w_i$ is unmatched, $s_i^1$ is to be matched with $d_i^1$, $s_i^2$ with $d_i^2$, and $s_i^3$ with $e_1$ and $e_2$. Student $s_i^3$ contains courses $c_j$ and $c_k$ in her preference list to make sure that a course-complete first-coalition-stable matching in $I$ corresponds to a maximal matching in $G$, so that if $s_i^3$ is matched with $e_1$ and $e_2$, some students from $S^1$ or $S^2$ are matched with $c_j$ and $c_k$. For each student $p_k\in P$, her preference list contains every course in $C'$ in increasing subscript order:
\begin{align*}
p_k: c_1 \succ_{p_k} c_2 \succ_{p_k} ... \succ_{p_k} c_{n_1}
\end{align*}

The preference list of every $c_j \in C'$ depends on the vertices adjacent to $u_j$. For every $w_i$ adjacent to $u_j$, $c_j$ includes in its preference list either $s_i^1$ or $s_i^2$, as well as $s_i^3$. Additionally, it includes every student in $P$. The order of all students is derived from $\succ_{MLS}$. If, for example, $u_j$ is adjacent to $w_3$, $w_4$ and $w_5$, the preference list of $c_j$ is given below. Here the superscript $l_i$ indicates a number that can be either $1$ or $2$. If $u_j$ is matched in $G$ then $c_j$ is matched with some student from $S^1$ or $S^2$, otherwise it is matched with some $p_k$.
\begin{align*}
c_j: s_3^{l_3} \succ_{c_j} s_4^{l_4} \succ_{c_j} s_5^{l_5} \succ_{c_j} s_3^3 \succ_{c_j} s_4^3 \succ_{c_j} s_5^3 \succ_{c_j} p_1 \succ_{c_j} p_2 \succ_{c_j} ... \succ_{c_j} p_{n_1}
\end{align*}

The preference lists of every course $d_i^1 \in D^1$ and $d_i^2 \in D^2$ are:
\begin{align*}
d_i^1&: s_i^1 \succ_{d_i^1} s_i^2 \\
d_i^2&: s_i^2 \succ_{d_i^2} s_i^3
\end{align*}

Finally, the preference lists of courses $e_1$ and $e_2$ contain every student from $S^3$ and the preference list of course $f$ contains every student from $S^1$ and $S^2$, in all cases in increasing subscript order; in the case of $f$'s preference list, every $s_i^1$ is followed by $s_i^2$. The students matched to $e_1$ and $e_2$ correspond to vertices in $G$ that are unmatched, while the students matched to $f$ correspond to vertices that are matched.
\begin{align*}
e_1&: s_1^3 \succ_{e_1} s_2^3 \succ_{e_1} ... \succ_{e_1} s_{n_2}^3 \\
e_2&: s_1^3 \succ_{e_2} s_2^3 \succ_{e_2} ... \succ_{e_2} s_{n_2}^3 \\
f&: s_1^1 \succ_f s_1^2 \succ_f s_2^1 \succ_f s_2^2 \succ_f ... \succ_f s_{n_2}^1 \succ_f s_{n_2}^2
\end{align*}

An example of this transformation is shown in Figure \ref{fig:ExampleTheoremPairSize}.

\begin{figure}[ht]
    \centering
    \begin{subfigure}[b]{0.35\linewidth}
        \centering
        \begin{tikzpicture}[main/.style = {node distance={25mm}, thick, draw, circle}]
            \node[main] (1) {$u_1$};
            \node[main] (2) [below of=1] {$u_2$};
            \node[main] (3) [above right of=1] {$w_1$};
            \node[main] (4) [below of=3] {$w_2$};
            \node[main] (5) [below of=4] {$w_3$};
            \draw (1) -- (3);
            \draw (1) -- (4);
            \draw (1) -- (5);
            \draw (2) -- (3);
            \draw (2) -- (4);
            \draw (2) -- (5);
        \end{tikzpicture}
    \end{subfigure}
    \hspace{0.05\linewidth}
    \begin{subfigure}[b]{0.55\linewidth}
        \centering
        \begin{align*}
        s_1^1&: c_1 \succ_{s_1^1} d_1^1 \succ_{s_1^1} f \\
        s_1^2&: c_2 \succ_{s_1^2} d_1^1 \succ_{s_1^2} d_1^2 \succ_{s_1^2} f \\
        s_1^3&: c_1 \succ_{s_1^3} c_2 \succ_{s_1^3} d_1^2 \succ_{s_1^3} e_1 \succ_{s_1^3} e_2 \\ \\
        c_1&: s_1^1 \succ_{c_1} s_2^1 \succ_{c_1} s_3^1 \succ_{c_1} s_1^3 \succ_{c_1} s_2^3 \succ_{c_1} s_3^3 \succ_{c_1} p_1 \\
        e_1&: s_1^3 \succ_{e_1} s_2^3 \succ_{e_1} s_3^3 \\
        f&: s_1^1 \succ_{f} s_1^2 \succ_{f} s_2^1 \succ_{f} s_2^2 \succ_{f} s_3^1 \succ_{f} s_3^2
        \end{align*}
    \end{subfigure}
    \caption{Example of an {\sc exact-mm} instance, with $K=2$, and some of the preference lists in the transformed instance of {\sc ca-ml}}
    \label{fig:ExampleTheoremPairSize}
\end{figure}

We first show that a course-complete matching in $I$ is student-complete, and vice versa.

\begin{lemma}
Let $M'$ be a matching in $I$. Then $M'$ is course-complete if and only if it is student-complete.
\label{lem:ThmPairSizeStabilityMaxSizeStudentCourseComplete}
\end{lemma}
\begin{proof}\renewcommand{\qedsymbol}{\ensuremath{\blacksquare}}
There are $n_1$ courses in $C'$, each with upper quota $1$ and number of credits $1$. There are also $n_2$ courses in each of $D^1$ and in $D^2$, each with upper quota $1$ and number of credits $2$. Finally, we have courses $e_1$ and $e_2$, each with upper quota $n_2-K$ and number of credits $1$, and $f$, with upper quota $K$ and number of credits $1$. Hence, a course-complete matching has $n_1 \cdot 1 \cdot 1 + 2\cdot n_2 \cdot 2 \cdot 1 + 2 \cdot (n_2-K) \cdot 1 + 1 \cdot K \cdot 1 = n_1 + 6n_2 - K$ credits.

Likewise, there are $n_2$ students in each of $S^1$, $S^2$ and $S^3$, each with two available credits, and $n_1-K$ students in $P$, each with one available credit. Thus, a student-complete matching has $3 \cdot n_2 \cdot 2 + (n_1-K) \cdot 1 = n_1 + 6n_2 - K$ number of credits.

Thus, a course-complete matching has the same number of credits as a student-complete matching.
\end{proof}

We now show that a maximal matching $M$ of size $K$ in $G$ implies a student-complete (course-complete by Lemma \ref{lem:ThmPairSizeStabilityMaxSizeStudentCourseComplete}) pair-size-stable matching $M'$ in $I$. For every pair $(u_j, w_i) \in M$, we add to $M'$ either the pairs $(s_i^1, c_j)$ and $(s_i^1, f)$, or the pairs $(s_i^2, c_j)$ and $(s_i^2, f)$, depending on whether $s_i^1$ or $s_i^2$ has $c_j$ in her preference list. If $(s_i^1, c_j) \in M'$, then we add $(s_i^2, d_i^1)$ to $M'$, otherwise we add $(s_i^1, d_i^1)$ to $M'$. We also add $(s_i^3, d_i^2)$ to $M'$. If $w_i$ has no partner in $M$, we add the following pairs to $M'$: $(s_i^1, d_i^1)$, $(s_i^2, d_i^2)$, $(s_i^3, e_1)$, and $(s_i^3, e_2)$. Finally, let $j_k$ ($1 \leq k \leq n_1-K$) be a sequence such that $w_{j_k}$ is unmatched in $M$. Then add $(p_k, c_{j_k})$ to $M'$ ($1 \leq k \leq n_1-K$). The matching $M'$ is course-complete, as each course in $C'$, $D^1$ and $D^2$ is matched with one student, each of $e_1$ and $e_2$ is matched with $n_2-K$ students, and $f$ is matched with $K$ students. It is also student-complete, because for every student $s \in S$, $T(s) = O(C_{M'}(s))$. We show below that $M'$ is pair-size-stable in $I$.

\begin{lemma}
$M'$ is a pair-size-stable matching in $I$.
\end{lemma}
\begin{proof}\renewcommand{\qedsymbol}{\ensuremath{\blacksquare}}
We prove the pair-size stability of $M'$ by showing that no student can belong to a size-blocking pair. Every $s_i^1$ is matched with either $f$ and some $c_j$, which is her most-preferred course, or with $d_i^1$. In either case, there is no course with which $s_i^1$ can form a size-blocking pair. Similarly, every $s_i^2$ is matched with either $f$ and some $c_k$, her most-preferred course, with $d_i^1$, or with $d_i^2$. If $s_i^2$ is matched with $f$ and $c_k$, or with $d_i^1$, the same reasoning applies as for $s_i^1$, and if she is matched with $d_i^2$, then $(s_i^1, d_i^1) \in M'$. As $s_i^1\succ_{d_i^1} s_i^2$, there is no size-blocking pair containing $s_i^2$.

We now consider $s_i^3$. This student is matched with either $d_i^2$ or with $e_1$ and $e_2$. In the first case, there is no course with at least two credits that $s_i^3$ prefers to $d_i^2$, so there is no possible size-blocking pair. In the second case, $(s_i^2, d_i^2) \in M'$, $s_i^2 \succ_{d_i^2} s_i^3$, and, because $M$ is maximal in $G$, $u_j$ and $u_k$ are matched in $M$ with some $w_p$ and $w_q$ respectively. Thus, there exist $s_p^{l_p},s_q^{l_q}\in S^1\cup S^2$ such that $(s_p^{l_p},c_j)\in M'$ and $(s_q^{l_q},c_k)\in M'$. Hence, $s_i^3$ cannot form a size-blocking pair with any other course. Finally, every $p_k \in P$ is matched with some $c_j \in C'$. Let $c_l$ be a course that $p_k$ prefers over $c_j$. Then $c_l$ is matched with either some $s_i^1$ or $s_i^2$, or with a $p_k'$ such that $p_k' \succ_{MLS} p_k$. Thus, $p_j$ cannot form a size-blocking pair with $c_l$.
\end{proof}

Now, assume that there exists a course-complete (student-complete) pair-size-stable matching $M'$ in $I$. We show how to construct a maximal matching $M$ of size $K$ in $G$. First, as $M'$ is student-complete, every $p_k \in P$ is matched with some $c_j \in C'$. As $M'$ is course-complete, the $K$ remaining courses in $C'$ are matched with students from $S^1$, $S^2$, or $S^3$. Similarly, as $M'$ is course-complete, $|S_M(f)| = q^+(f) = K$, that is, course $f$ fills its capacity. Every student who has $f$ in her preference list belongs to either $S^1$ or $S^2$, and because $M'$ is student-complete every student matched to $f$ must be matched to some $c_j$, so the $K$ courses from $C'$ that are not matched with students from $P$ are matched with students from $S^1$ and $S^2$. We now show that for each subset $S_i$ ($1\leq i\leq n_2$), either $s_i^1$ or $s_i^2$ is matched with some $c_j \in C'$, but not both.

\begin{lemma}
If $(s_i^1, c_j) \in M'$ or $(s_i^2, c_j) \in M'$, then $M'$ includes no other pair between a member of $S_i$ and a member of $C'$.
\end{lemma}
\begin{proof}\renewcommand{\qedsymbol}{\ensuremath{\blacksquare}}
If $(s_i^1,c_j)\in M'$, then $(s_i^2, d_i^1)$ must be in $M'$, otherwise $d_i^1$ has no students and $M'$ is not course-complete. Thus, $(s_i^3, d_i^2) \in M'$, or else $d_i^2$ has no students. If $(s_i^2,c_k)\in M'$, then, by the same reasoning, $(s_i^2, d_i^1) \in M'$ and $(s_i^3, d_i^2) \in M'$.
\end{proof}

Thus, we can build a matching $M$ in $G$ of size $K$ by adding the pair $(u_j, w_i)$ to $M$ for every pair of the form $(s_i^{l_i}, c_j) \in M'$, of which there are exactly $K$. For every $S_i$ such that no $s_i^{l_i}\in S_i$ is matched to some $c_j \in C'$, we have that $\{(s_i^1, d_i^1),(s_i^2, d_i^2), (s_i^3, e_1),(s_i^3, e_2)\}\subseteq M'$. We now show that $M$ is a maximal matching in $G$.

\begin{lemma}
The matching $M$ in $G$, of size $K$, is maximal.
\end{lemma}
\begin{proof}\renewcommand{\qedsymbol}{\ensuremath{\blacksquare}}
Looking for a contradiction, assume that $M$ is not maximal, so that there exists unmatched $u_j, w_i$ that are adjacent to each other in $G$. This implies that in $M'$, $c_j$ is matched with some $p_k$ and $s_i^3$ is matched with $e_1$ and $e_2$. However, $s_i^3$ prefers $c_j$ over each of $e_1$ and $e_2$, and $c_j$ prefers $s_i^3$ over $p_k$. Thus, $(s_i^3, c_j)$ is a size-blocking pair, a contradiction to the pair-size stability of $M'$ in $I$. Therefore, $M$ is maximal in $G$.
\end{proof}

Hence, there exists a maximal matching of size $K$ in $G$ if and only if there exists a course-complete (student-complete) pair-size-stable matching in $I$. As {\sc exact-mm} is NP-complete, determining the existence of a course-complete (student-complete) pair-size-stable matching is NP-hard.
\end{proof}

We can modify this transformation to show that determining the existence of a course-complete (but not necessarily student-complete) pair-size-stable matching is NP-hard even if all conditions from the theorem statement hold and the length of the preference list of every student is bounded. We do so by removing every student in $P$, and creating a new $p_j$ student for every $c_j \in C'$. The preference list of this $p_j$ contains $c_j$ only, and these $p_j$ students are at the end of the master list of students. We would then have that the maximum length of each student's preference list is $5$.


\subsection{With lower quotas}
\label{subsec:maximumLowerQuotas}
A very natural requirement, present in many universities, is that courses have minimum numbers of students that they require in order to run, which are represented by lower quotas.  We have not considered this restriction until now, but in this section we now define an extension to {\sc ca} involving lower quotas.

A course $c_j$ is said to have a \emph{lower quota}, denoted by $q^-(c_j)$, if in a matching $M$ the number of students allocated to $c_j$ must be at least the lower quota, i.e., $|S_M(c_j)| \geq q^-(c_j)$. When courses have lower quotas, there are two variations to consider. In the first one, every course in the {\sc ca} instance must have a number of students at least as large as its lower quota. The goal is to find a stable matching, where every course is assigned a number of students between its lower and upper quotas, for any given definition of stability. This problem is called {\sc ca-x-lq-nc-find}, where {\sc x} denotes the definition of stability. In the case of {\sc ca-ml}, the problem is called {\sc ca-x-ml-lq-nc-find}.

The second variation models the scenario where in a given matching $M$, if course $c_j$ does not have at least as many students as its lower quota, then it may be closed, and no students are matched to it in this case. In this setting an empty matching, where every course is closed, is trivially stable, so our goal is to find a maximum size stable matching. This problem is called {\sc ca-x-lq-cl-max}, where x denotes the definition of stability. In the case of {\sc ca-ml}, the problem is called {\sc ca-x-ml-lq-cl-max}. In this variation we do not allow for the existence of a coalition of several students and a closed course where the students have an incentive to demand that this course be opened (as in \cite{biro_college_2010}), because in real-life settings lower quotas are usually too high for students to coordinate in such a fashion.

From Theorems \ref{thm:firstCoalitionStabilityMaxSize} and \ref{thm:pairSizeStabilityMaxSize} and Corollary \ref{cor:coalitionStabilityMaxSize} we immediately obtain the following two results for {\sc ca-ml} with lower quotas both without and with course closures.

\begin{observation}
Given an instance $I$ of {\sc ca-ml} with lower quotas, each of {\sc ca-ps-ml-lq-nc-find}, {\sc ca-c-ml-lq-nc-find} and {\sc ca-fc-ml-lq-nc-find} is NP-hard, even if every student takes up to two credits, every course has one or two credits, the length of the students' preference lists is bounded, and there are no downward-feasible constraints.
\label{obs:lowerQuotasNoClosuresThreeDefinitionsStability}
\end{observation}
\begin{proof}
Modify the instances from Theorems \ref{thm:firstCoalitionStabilityMaxSize} and \ref{thm:pairSizeStabilityMaxSize} and Corollary \ref{cor:coalitionStabilityMaxSize} so that the lower quota of every course is the same as the upper quota. Then a stable matching respecting lower quotas exists if and only if a course-complete stable matching exists, for all three definitions of stability.
\end{proof}

\begin{observation}
{\sc ca-ps-ml-lq-cl-max}, {\sc ca-c-ml-lq-cl-max} and {\sc ca-fc-ml-lq-cl-max} is NP-hard, even if every student takes up to two credits, every course has one or two credits, the length of the students' preference lists is bounded, and there are no downward-feasible constraints.
\label{obs:lowerQuotasClosuresThreeDefinitionsStability}
\end{observation}
\begin{proof}
Under pair-size stability, coalition stability or first-coalition stability, finding a maximum size stable matching is NP-hard without lower quotas (see Theorem \ref{thm:pairSizeStabilityMaxSize}, Corollary \ref{cor:coalitionStabilityMaxSize} and Theorem \ref{thm:firstCoalitionStabilityMaxSize} respectively), so the problem remains hard with lower quotas.
\end{proof}

Similarly, the following result is an immediate consequence of Theorem \ref{thm:pairStabilityNPCompleteness}:
\begin{observation}
Given an instance of {\sc ca} with lower quotas and no course closures, determining the existence of a pair-stable matching is NP-complete.  That is, {\sc ca-p-lq-nc-find} is NP-hard.
\end{observation}
However, in the case of master lists, we know that every instance of {\sc ca-ml} has a unique pair-stable matching. Thus, in order to determine whether an instance of {\sc ca-ml} with lower quotas has a pair-stable matching that respects lower quotas, we just need to run Algorithm \ref{alg:findMaster} and check whether in the resulting matching every course has enough students to fill its lower quota. Thus, we obtain the following Observation.

\begin{observation}
{\sc ca-p-ml-lq-nc-find} is solvable in polynomial time.
\label{obs:lowerQuotasNoClosuresPairStability}
\end{observation}

On the other hand, if we allow for course closures, {\sc ca-p-ml-lq-cl-max} is NP-hard.

\begin{theorem}
It is NP-complete to determine if an instance of {\sc ca-ml} with lower quotas and closures has a student-complete pair-stable matching, even if both students and courses are ordered according to master lists, every student takes up to two credits, every course has one or two credits and lower quota $1$ or $2$, and there are no downward-feasible constraints. Hence, it follows that {\sc ca-p-ml-lq-cl-max} is NP-hard.
\label{thm:lowerQuotasClosuresPairStability}
\end{theorem}
\begin{proof}
Verifying whether a matching is student-complete and pair-stable, and respects lower quotas, can be done in polynomial time, so the decision problem belongs to NP. We prove NP-completeness by reducing from {\sc exact-mm}, which is NP-complete, even if $G$ is bipartite and all vertices in one of the sides have degree $2$ \cite{manlove_2002_hard}\citep[][page 133]{manlove_algorithmics_2013}. Let $G = (U, W, E)$ be a graph and $K$ be an integer representing an instance of {\sc exact-mm}, where $U$ and $W$ are the two sets of vertices in the bipartition, $|U| = n_1$ and $|W| = n_2$, and $K$ is the target size of a maximal matching. Without loss of generality assume that every vertex in $W$ has degree $2$. We build an instance $I$ of {\sc ca-ml-lq-cl}, with students $S$ and courses $C$, as follows.

For every $w_i \in W$, we create the following students: $s_i^1$, $s_i^2$, $s_i^3$, and $s_i^4$, with $T(s_i^1) = T(s_i^2) = T(s_i^3) = T(s_i^4) = 2$. We also create courses $d_i^1$ and $d_i^2$, with $O(d_i^1) = O(d_i^2) = 2$, as well as courses $b_i^1$ and $b_i^2$, with $O(b_i^1) = O(b_i^2) = 1$.
Create the following subsets: $S_i=\{s_i^1,s_i^2,s_i^3,s_i^4\}$ ($1\leq i\leq n_2$), $D_i=\{d_i^1,d_i^2\}$ ($1\leq i\leq n_2$), $S^r = \{s_i^r : 1\leq i\leq n_2\}$ ($1\leq r\leq 4$), $D^r = \{d_i^r : 1\leq i\leq n_2\}$ ($1\leq r\leq 2$) and $B=\{b_i^r : 1\leq i\leq n_2\wedge 1\leq r\leq 2\}$.

For every $u_j \in U$, we create a course $c_j$, with $O(c_j) = 1$. Let $C'=\{c_j : 1\leq j\leq n_1\}$. Additionally, we create two courses $e$ and $f$, with $O(e) = O(f) = 2$. The upper quota of every course in $C'\cup D^2$ is $1$, the upper quota of every course in $B\cup D^1\cup \{f\}$ is $2$, and the upper quota of $e$ is $n_2-K$.
The lower quota of every course in $C'\cup D^2
\cup \{e\}$ is $0$, 
whilst the lower quota of every course in $B\cup D^1\cup \{f\}$ is $2$.

We also create $K$ students $P=\{p_k : 1\leq k\leq K\}$ such that, for every $p_k\in P$, $T(p_k) = 1$, and we create two students $q$ and $r$ such that $T(q) = T(r) = 2$. 

The master list of students ($\succ_{MLS}$) is
\begin{multline}
s_1^1 \succ_{MLS} s_1^4 \succ_{MLS} s_1^2 \succ_{MLS} s_2^1 \succ_{MLS} s_2^4 \succ_{MLS} s_2^2 \succ_{MLS} ... \succ_{MLS} s_{n_2}^1 \succ_{MLS} \\ s_{n_2}^4 \succ_{MLS} s_{n_2}^2 \succ_{MLS} s_1^3 \succ_{MLS} s_2^3 \succ_{MLS} ... \succ_{MLS} s_{n_2}^3 \succ_{MLS} \\ p_1 \succ_{MLS} ... \succ_{MLS} p_K \succ_{MLS} q \succ_{MLS} r
\end{multline}

That is, the first students belong to subsets $S^1$, $S^4$ and $S^2$, such that every $s_i^1$ is followed by $s_i^4$ and then by $s_i^2$. Then we place students from $S^3$, and after them students from $P$, in all cases in increasing subscript order. Finally, we place students $q$ and $r$. The master list of courses ($\succ_{MLC}$) is
\begin{multline}
b_1^1 \succ_{MLC} b_1^2 \succ_{MLC} b_2^1 \succ_{MLC} b_2^2 \succ_{MLC} ... \succ_{MLC} b_{n_2}^1 \succ_{MLC} b_{n_2}^2 \succ_{MLC} \\ d_1^1 \succ_{MLC} d_1^2 \succ_{MLC} d_2^1 \succ_{MLC} d_2^2 \succ_{MLC} ... \succ_{MLC} d_{n_2}^1 \succ_{MLC} d_{n_2}^2 \succ_{MLC} \\ c_1 \succ_{MLC} c_2 \succ_{MLC} ... \succ_{MLC} c_{n_1} \succ_{MLC} e \succ_{MLC} f
\end{multline}

That is, the first courses are those from $B$, such that every $b_i^1$ is followed by $b_i^2$, then the next courses are from $D^1\cup D^2$, such that every $d_i^1$ is followed by $d_i^2$; after these, the courses in $C'$ follow, in all cases in increasing subscript order. Finally we have $e$ followed by $f$.

Let $w_i \in W$ be a vertex from $G$, adjacent to $u_j$ and $u_k$ (with $j < k$). We build the preference lists of $s_i^1$, $s_i^2$ and $s_i^3$ as follows.
\begin{align*}
s_i^1&: b_i^1 \succ_{s_i^1} d_i^1 \succ_{s_i^1} c_j \\
s_i^2&: b_i^2 \succ_{s_i^2} d_i^1 \succ_{s_i^2} d_i^2 \succ_{s_i^2} c_k \\
s_i^3&: d_i^2 \succ_{s_i^3} c_j \succ_{s_i^3} c_k \succ_{s_i^3} e \\
s_i^4&: d_i^1
\end{align*}

If in the {\sc exact-mm} instance $w_i$ is matched with $u_j$, then $s_i^1$ is to be matched with $c_j$ and $b_i^1$, $s_i^2$ with $d_i^1$, $s_i^3$ with $d_i^2$, and $s_i^4$ with $d_i^1$. Similarly if $w_i$ is matched with $u_k$, then $s_i^1$ is to be matched with $d_i^1$, $s_i^2$ with $c_k$ and $b_i^2$, $s_i^3$ with $d_i^2$ and $s_i^4$ with $d_i^1$. If, however, $w_i$ is unmatched, $s_i^1$ is to be matched with $d_i^1$, $s_i^2$ with $d_i^2$, $s_i^3$ with $e$, and $s_i^4$ with $d_i^1$. Student $s_i^3$ contains courses $c_j$ and $c_k$ in her preference list to make sure that a course-complete pair-stable matching in $I$ corresponds to a maximal matching in $G$, so that if $s_i^3$ is matched with $e$, some students from $S^1$ or $S^2$ are matched with $c_j$ and $c_k$. The preference lists of students in $P$ includes every course in $B$.
\begin{align*}
p_k: b_1^1 \succ_{p_k} b_1^2 \succ_{p_k} b_2^1 \succ_{p_k} b_2^2 \succ_{p_k} ... \succ_{p_k} b_{n_2}^1 \succ_{p_k} b_{n_2}^1
\end{align*}

Students from $P$ are matched to the courses in $B$ that have some $s_i^l$ matched to them, which correspond to vertices $w_i \in G$ that are matched in $G$. The preference lists of $q$ and $r$ are:
\begin{align*}
q&: e \succ_q f \\
r&: f
\end{align*}

The preference list of every $c_j \in C'$ depends on the vertices adjacent to $u_j$. For every $w_i$ adjacent to $u_j$, $c_j$ includes in its preference list either $s_i^1$ or $s_i^2$, as well as $s_i^3$. The order of all students is derived from the master list. If, for example, $u_j$ is adjacent to $w_3$, $w_4$ and $w_5$, the preference list of $c_j$ is given below. Here the superscript $l_i$ indicates a number that can be either $1$ or $2$. If $u_j$ is matched in $G$ then $c_j$ is matched with some student from $S^1$ or $S^2$, otherwise it is unmatched.
\begin{align*}
c_j: s_3^{l_3} \succ_{c_j} s_4^{l_4} \succ_{c_j} s_5^{l_5} \succ_{c_j} s_3^3 \succ_{c_j} s_4^3 \succ_{c_j} s_5^3
\end{align*}

The preference lists of every $b_i^1$ and $b_i^2$ in $B$ include $s_i^1$ in the case of $b_i^1$, and $s_i^2$ in the case of $b_i^2$. Then, we have every student in $P$.
\begin{align*}
b_i^1&: s_i^1 \succ_{b_i^1} p_1 \succ_{b_i^1} p_2 \succ_{b_i^1} ... \succ_{b_i^1} p_K \\
b_i^2&: s_i^2 \succ_{b_i^2} p_1 \succ_{b_i^2} p_2 \succ_{b_i^2} ... \succ_{b_i^2} p_K
\end{align*}

The preference lists of every $d_i^1 \in D^1$, $d_i^2 \in D^2$ and of $f$ are:
\begin{align*}
d_i^1&: s_i^1 \succ_{d_i^1} s_i^4 \succ_{d_i^1} s_i^2 \\
d_i^2&: s_i^2 \succ_{d_i^2} s_i^3 \\
f&: q \succ_f r
\end{align*}

Finally, the preference list of $e$ contains every student from $S^3$ in increasing subscript order, and then $q$. The students from $S^3$ matched to $e$ correspond to vertices in $G$ that are unmatched.
\begin{align*}
e: s_1^3 \succ_e s_2^3 \succ_e ... \succ_e s_{n_2}^3 \succ_e q
\end{align*}

An example of this transformation is shown in Figure \ref{fig:ExampleTheoremPairStabilityLQ}.

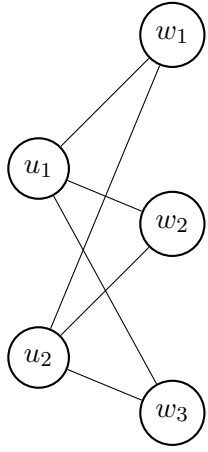
\begin{figure}[ht]
    \centering
    \begin{subfigure}[b]{0.35\linewidth}
        \centering
        \begin{tikzpicture}[main/.style = {node distance={25mm}, thick, draw, circle}]
            \node[main] (1) {$u_1$};
            \node[main] (2) [below of=1] {$u_2$};
            \node[main] (3) [above right of=1] {$w_1$};
            \node[main] (4) [below of=3] {$w_2$};
            \node[main] (5) [below of=4] {$w_3$};
            \draw (1) -- (3);
            \draw (1) -- (4);
            \draw (1) -- (5);
            \draw (2) -- (3);
            \draw (2) -- (4);
            \draw (2) -- (5);
        \end{tikzpicture}
    \end{subfigure}
    \hspace{0.05\linewidth}
    \begin{subfigure}[b]{0.55\linewidth}
        \centering
        \begin{align*}
        s_1^1&: b_1^1 \succ_{s_1^1} d_1^1 \succ_{s_1^1} c_1 \\
        s_1^2&: b_1^2 \succ_{s_1^2} d_1^1 \succ_{s_1^2} d_1^2 \succ_{s_1^2} c_2 \\
        s_1^3&: d_1^2 \succ_{s_1^3} c_1 \succ_{s_1^3} c_2 \succ_{s_1^3} e \\
        s_1^4&: d_1^1 \\
        p_1&: b_1^1 \succ_{p_1} b_1^2 \succ_{p_1} b_2^1 \succ_{p_1} b_2^2 \succ_{p_1} b_3^1 \succ_{p_1} b_3^2 \\ \\
        c_1&: s_1^1 \succ_{c_1} s_2^1 \succ_{c_1} s_3^1 \succ_{c_1} s_1^3 \succ_{c_1} s_2^3 \succ_{c_1} s_3^3 \\
        e&: s_1^3 \succ_{e} s_2^3 \succ_{e} s_3^3 \succ_e q \\
        f&: q \succ_f r
        \end{align*}
    \end{subfigure}
    \caption{Example of an {\sc exact-mm} instance, with $K=2$, and some of the preference lists in the transformed instance of {\sc ca-ml-lq-cl}.}
    \label{fig:ExampleTheoremPairStabilityLQ}
\end{figure}

We now show that a maximal matching $M$ of size $K$ in $G$ implies a student-complete pair-stable matching $M'$ in $I$ that respects lower quotas. For every pair $(u_j, w_i) \in M$, we add to $M'$ either the pairs $(s_i^1, c_j)$ and $(s_i^1, b_i^1)$, or the pairs $(s_i^2, c_j)$ and $(s_i^2, b_i^2)$, depending on whether $s_i^1$ or $s_i^2$ has $c_j$ in her preference list. If $(s_i^1, c_j) \in M'$, then we add $(s_i^2, d_i^1)$ to $M'$, otherwise we add $(s_i^1, d_i^1)$ to $M'$. We also add $(s_i^3, d_i^2)$ and $(s_i^4, d_i^1)$ to $M'$. If $w_i$ has no partner in $M$, we add the following pairs to $M'$: $(s_i^1, d_i^1)$, $(s_i^2, d_i^2)$, $(s_i^3, e)$, and $(s_i^4, d_i^1)$. Additionally, let $j_k$ ($1 \leq k \leq K$) be a sequence such that $w_{j_k}$ is matched in $M$. For each $k$ ($1\leq k\leq K$), if $(s_i^1, c_{j_k}) \in M'$ for some $s_i^1$, add $(p_{j_k}, b_i^1)$ to $M'$, otherwise add $(p_{j_k}, b_i^2)$ to $M'$. Finally, add pairs $(q, f)$ and $(r, f)$ to $M'$. We close every course in $B$ that has no students matched to it. The matching $M'$ is student-complete, because for every student $s \in S$, $T(s) = O(C_{M'}(s))$. It also respects the lower quotas of every course that appears in $M'$: in particular, course $f$ has two students, and every open course in $B$ and $D^1$ has two students also. We show below that $M'$ is pair-stable.

\begin{lemma}
$M'$ is a pair-stable matching in $I$.
\end{lemma}
\begin{proof}\renewcommand{\qedsymbol}{\ensuremath{\blacksquare}}
We prove the pair stability of $M'$ by showing that no student can belong to a blocking pair. Every $s_i^1$ is matched with either $b_i^1$ (which is her first choice) and some $c_j$, or with $d_i^1$. In the first case, $s_i^1$ cannot block with $d_i^1$, as she cannot drop enough credits to do so. In the second case, $b_i^1$ is closed, so she cannot form a blocking pair with $b_i^1$. Similarly, every $s_i^2$ is matched with either (i) $b_i^2$ (her first choice) and some $c_k$, (ii) with $d_i^1$, or (iii) with $d_i^2$. If $s_i^2$ is matched with $b_i^2$ and $c_k$, or with $d_i^1$, the same reasoning applies as for $s_i^1$, and if she is matched with $d_i^2$, then $(s_i^1, d_i^1) \in M'$ and $(s_i^4, d_i^1) \in M'$, so there is no blocking pair containing $s_i^2$. Student $s_i^4$ is always matched with the only course in her preference list, $d_i^1$.

We now consider $s_i^3$. This student is matched with either $d_i^2$, her most-preferred course, or with $e$. In the second case, $(s_i^2, d_i^2) \in M'$ and $w_i$ is unmatched in $M$. As $M$ is maximal in $G$, $u_j$ and $u_k$ are matched in $M$ with some $w_p$ and $w_q$ respectively. Thus, there exist $s_p^{l_p},s_q^{l_q}\in S^1\cup S^2$ such that $(s_p^{l_p},c_j)\in M'$ and $(s_q^{l_q},c_k)\in M'$.  Each of these courses prefers its assigned student in $M'$ to $s_i^3$. Hence, $s_i^3$ cannot form a blocking pair with any other course. Moving on to students in $P$, every $p_k \in P$ is matched with some $b_i^r\in B$, where $r\in \{1,2\}$. Let $b_j^s$ be an open course that $p_k$ prefers over $b_i^r$. Then $b_j^s$ is matched with both some $s_l^1$ or $s_l^2$, and also with a $p_k'$ such that $p_k' \succ_{MLS} p_k$. Thus, $p_k$ cannot form a blocking pair with $b_i^l$. Finally, $r$ is matched with the only course in her preference list, $f$, and $q$ is matched with $f$, too. Even though $q$ prefers $e$ over $f$, $e$ has filled its capacity, and every student enrolled in $e$ belongs to $S^3$, all of which appear in the master list before $q$. Therefore, no student belongs to a blocking pair.
\end{proof}

Now, assume that there exists a student-complete pair-stable matching $M'$ in $I$. We show how to construct a maximal matching $M$ of size $K$ in $G$. First, as $M'$ is student-complete, $r$ is matched with $f$. The minimum quota of $f$ is $2$, so we have that $(q, f) \in M'$ too. However, $q$ prefers $e$ over $f$. Thus, because $M'$ is pair-stable, $e$ must have $n_2-K$ students assigned to it that are better than $q$. 
Hence, there are $n_2-K$ students from $S^3$ matched to $e$ in $M'$. As $M'$ is student-complete, $(s_i^4, d_i^1)\in M'$, for all $i$ ($1\leq i\leq n_2$) and thus every course in $D^1$ is open. In order to preserve pair stability, for each $s_i^3$ matched to $e$, $M'$ contains the pairs $(s_i^1, d_i^1)$ and $(s_i^2, d_i^2)$.

As $M'$ is student-complete, every student $p_k\in P$ (there are $K$ such students) is matched to some $b_i^r\in B$, where $r\in \{1,2\}$. In order to satisfy the lower quota of $b_i^r$, and by the pair stability of $M'$, $(s_i^r, b_i^r) \in M'$ also. As $M'$ is student-complete, $(s_i^r,c_j)\in M'$ for some $c_j\in C'$. As the lower quota of $d_i^1$ is $2$, and we have already established that $d_i^1$ is open, $M'$ must include the pair $(s_i^{3-r}, d_i^1)$. This implies that $(s_i^3, d_i^2) \in M'$ by pair stability. Hence, for each $i$ ($1 \leq i \leq n_2$), if $(s_i^1,c_j)\in M$ for some $c_j\in C'$ then $(s_i^2,d_i^1)\in M'$, whilst if $(s_i^2,c_k)\in M'$ for some $c_k\in C'$ then $(s_i^1,d_i^1)\in M'$. Thus, we can build a matching $M$ in $G$ by adding the pair $(u_j, w_i)$ to $M$ for every pair of the form $(s_i^r, c_j) \in M'$. It follows by the fact that $|S_M(e)|=n_2-K$ that $|M|=K$. We now show that $M$ is a maximal matching in $G$.

\begin{lemma}
The matching $M$ in $G$, of size $K$, is maximal.
\end{lemma}
\begin{proof}\renewcommand{\qedsymbol}{\ensuremath{\blacksquare}}
Looking for a contradiction, assume that $M$ is not maximal, so that there exists unmatched $u_j, w_i$ that are adjacent to each other in $G$. This implies that in $M'$, $c_j$ is unmatched and $s_i^3$ is matched with $e$. However, $s_i^3$ prefers $c_j$ over $e$, whilst $c_j$ is unmatched, has lower quota $0$, and finds $s_i^3$ acceptable. Thus, $(s_i^3, c_j)$ is a blocking pair of $M'$, a contradiction to the pair stability of $M'$ in $I$. Therefore, $M$ is maximal in $G$.
\end{proof}

Hence, there exists a maximal matching of size $K$ in $G$ if and only if there exists a student-complete first-coalition-stable matching in $I$, respecting lower quotas of every open course. 
As {\sc exact-mm} is NP-complete, it follows that {\sc ca-p-ml-lq-cl-max} is NP-hard.
\end{proof}

\section{Conclusion}
\label{sec:conclusion}
We have introduced a new model of the {\sc course allocation} problem, which allows for courses with different numbers of credits and constraints such as courses that run at overlapping times. We have considered four natural definitions of stability in this context, and for each one we studied the problem of determining whether a matching is stable and finding a stable matching if one exists. Additionally, we studied the problem with master lists, and considered the problems of finding a stable matching and a maximum size stable matching, both in the absence of, and in the presence of, lower quotas.

The main open problem remaining is to establish the computational complexity of finding a coalition-stable matching under arbitrary preferences, and indeed to determine whether one always exists. Another possible research direction is to study stability under additively separable preferences, where students assign a utility to a course and prefer subsets of courses where the sum of utilities is higher \cite{cechlarova_pareto_2018}. A restriction of additively separable preferences is lexicographic preferences, where students compare two bundles of courses by first checking the first course of each bundle, then the second course, and so on \cite{cechlarova_pareto_2018}. If we also require that students participating in a blocking (according to additively separable or lexicographic preferences) coalition do not reduce their total number of credits, then these definitions of stability are less strict than first-coalition stability but more so than coalition stability.

Other possible future research directions include considering problems involving finding stable matchings with the minimum number of blocking pairs/coalitions when there is no stable matching, developing approximation algorithms for finding a maximum size stable matching, such as the 3-approximation algorithm for the problem of finding a maximum size occupancy-stable (equivalent to pair-size-stable) matching by Balasundaram et al.\ \cite{balasundaram_stability_2025}, or investigating problems involving modifying course upper quotas to ensure the existence of a stable matching, following the work of Nguyen and Vohra \cite{nguyen_2018_near}.

\section*{Acknowledgements}
The authors would like to thank Zoe Graves and Matt Wierzbicki of the School of Law, University of Glasgow, for inspiring the ideas behind our model for {\sc ca}; Mathijs Barkel and Fraser Paterson for valuable discussions concerning algorithms for {\sc ca-c-ml-max}; and the SAGT reviewers for their valuable comments and helpful suggestions. José Rodríguez is supported by a studentship from the University of Glasgow, College of Science and Engineering. The authors have no competing interests to declare that are relevant to the content of this article. For the purpose of open access, the authors have applied a Creative Commons Attribution (CC BY) licence to any Author Accepted Manuscript version arising from this submission.

\bibliographystyle{abbrv}
\bibliography{bibliography}

@inproceedings{rodriguez_course_2025,
	series = {Lecture {Notes} in {Computer} {Science}},
	title = {Course Allocation with Credits via Stable Matching},
	booktitle = {Proceedings of SAGT 2025: the 18th International Symposium on Algorithmic Game Theory},
	publisher = {Springer},
	author = {José Rodríguez and David Manlove},
	year = {2025},
    volume = {15953},
	pages = {286-303},
}

@article{shriya_approaches_2026,
   author = {Shriya, Shimona and Ranjan, Prabhat},
   title = {Approaches to assigning courses to students: a systematic review of existing literature and emerging topics},
   journal = {OPSEARCH, \emph{to appear}},
   year = {2026}
}

@book{manlove_algorithmics_2013,
	series = {Series on {Theoretical} {Computer} {Science}},
	title = {Algorithmics of {Matching} {Under} {Preferences}},
	volume = {2},
	isbn = {978-981-4425-24-7 978-981-4425-25-4},
	url = {http://www.worldscientific.com/worldscibooks/10.1142/8591},
	language = {en},
	urldate = {2021-10-20},
	publisher = {World Scientific},
	author = {Manlove, David F},
	year = {2013},
	doi = {10.1142/8591},
}

@article{gale_college_1962,
	title = {College {Admissions} and the {Stability} of {Marriage}},
	volume = {69},
	issn = {0002-9890},
	url = {https://doi.org/10.1080/00029890.1962.11989827},
	doi = {10.1080/00029890.1962.11989827},
	number = {1},
	urldate = {2021-11-03},
	journal = {The American Mathematical Monthly},
	author = {Gale, D. and Shapley, L. S.},
	year = {1962},
	pages = {9--15},
}

@article{abraham_stable_2008,
	author = {Abraham, David J. and Levavi, Ariel and Manlove, David F. and O’Malley, Gregg},
    title = {The {Stable} {Roommates} {Problem} with {Globally}-{Ranked} {Pairs}},
	journal = {Internet Mathematics},
    volume = {5},
    number = {4},
    pages = {493-515},
    year = {2008}
}

@article{roth_deferred_2008,
	title = {Deferred acceptance algorithms: history, theory, practice, and open questions},
	volume = {36},
	issn = {1432-1270},
	shorttitle = {Deferred acceptance algorithms},
	url = {https://doi.org/10.1007/s00182-008-0117-6},
	doi = {10.1007/s00182-008-0117-6},
	language = {en},
	number = {3},
	urldate = {2023-09-18},
	journal = {International Journal of Game Theory},
	author = {Roth, Alvin E.},
	year = {2008},
	pages = {537--569},
}

@article{stroh-maraun_weighted_2024,
	title = {Weighted school choice problems and the weighted top trading cycles mechanism},
	volume = {132},
	issn = {0165-4896},
	url = {https://www.sciencedirect.com/science/article/pii/S0165489624000817},
	doi = {10.1016/j.mathsocsci.2024.09.001},
	urldate = {2024-09-30},
	journal = {Mathematical Social Sciences},
	author = {Stroh-Maraun, Nadja},
	year = {2024},
	pages = {49--56},
}

@article{cechlarova_pareto_2018,
	title = {Pareto optimal matchings of students to courses in the presence of prerequisites},
	volume = {29},
	issn = {1572-5286},
	url = {https://www.sciencedirect.com/science/article/pii/S157252861730141X},
	doi = {10.1016/j.disopt.2018.04.004},
	urldate = {2024-12-03},
	journal = {Discrete Optimization},
	author = {Cechlárová, Katarína and Klaus, Bettina and Manlove, David F.},
	year = {2018},
	pages = {174--195},
}

@article{budish_course_2017,
	title = {Course {Match}: {A} {Large}-{Scale} {Implementation} of {Approximate} {Competitive} {Equilibrium} from {Equal} {Incomes} for {Combinatorial} {Allocation}},
	volume = {65},
	issn = {0030-364X},
	shorttitle = {Course {Match}},
	url = {https://pubsonline.informs.org/doi/abs/10.1287/opre.2016.1544},
	doi = {10.1287/opre.2016.1544},
	number = {2},
	urldate = {2024-12-05},
	journal = {Operations Research},
	author = {Budish, Eric and Cachon, Gérard P. and Kessler, Judd B. and Othman, Abraham},
	year = {2017},
	pages = {314--336},
}

@article{mcdermid_keeping_2010,
	title = {Keeping partners together: algorithmic results for the hospitals/residents problem with couples},
	volume = {19},
	issn = {1573-2886},
	shorttitle = {Keeping partners together},
	url = {https://doi.org/10.1007/s10878-009-9257-2},
	doi = {10.1007/s10878-009-9257-2},
	language = {en},
	number = {3},
	urldate = {2024-12-06},
	journal = {Journal of Combinatorial Optimization},
	author = {McDermid, Eric J. and Manlove, David F.},
	year = {2010},
	pages = {279--303},
    note = {See \url{https://www.dcs.gla.ac.uk/~davidm/HRS-errata.pdf} for an erratum},
}

@article{hoyer_stability_2020,
	title = {Stability in {Weighted} {College} {Admissions} {Problems}},
	url = {https://ideas.repec.org//p/pdn/dispap/63.html},
	language = {en},
	urldate = {2025-01-13},
	journal = {Working Papers Dissertations},
	author = {Hoyer, Britta and Stroh-Maraun, Nadja},
	year = {2020},
	note = {Number: 63
Publisher: Paderborn University, Faculty of Business Administration and Economics},
}

@inproceedings{aziz_stability_2018,
	title = {Stability and {Pareto} {Optimality} in {Refugee} {Allocation} {Matchings}},
	booktitle = {Proceedings of the AAMAS 2018: the 17th {International} {Conference} on {Autonomous} {Agents} and {MultiAgent} {Systems}},
	publisher = {International Foundation for Autonomous Agents and Multiagent Systems},
	author = {Aziz, Haris and Chen, Jiayin and Gaspers, Serge and Sun, Zhaohong},
	year = {2018},
	pages = {964--972},
}

@article{biro_college_2010,
	title = {The {College} {Admissions} problem with lower and common quotas},
	volume = {411},
	issn = {0304-3975},
	url = {https://www.sciencedirect.com/science/article/pii/S0304397510002860},
	doi = {10.1016/j.tcs.2010.05.005},
	number = {34},
	urldate = {2025-01-14},
	journal = {Theoretical Computer Science},
	author = {Biró, Péter and Fleiner, Tamás and Irving, Robert W. and Manlove, David F.},
	month = jul,
	year = {2010},
	pages = {3136--3153},
}

@article{biro_matching_2014,
	series = {Combinatorial {Optimization}},
	title = {Matching with sizes (or scheduling with processing set restrictions)},
	volume = {164},
	issn = {0166-218X},
	url = {https://www.sciencedirect.com/science/article/pii/S0166218X11004215},
	doi = {10.1016/j.dam.2011.11.003},
	urldate = {2025-01-14},
	journal = {Discrete Applied Mathematics},
	author = {Biró, Péter and McDermid, Eric},
	year = {2014},
	pages = {61--67},
}

@article{irving_stable_2008,
	title = {The stable marriage problem with master preference lists},
	volume = {156},
	issn = {0166-218X},
	url = {https://www.sciencedirect.com/science/article/pii/S0166218X0800022X},
	doi = {10.1016/j.dam.2008.01.002},
	number = {15},
	urldate = {2025-02-03},
	journal = {Discrete Applied Mathematics},
	author = {Irving, Robert W. and Manlove, David F. and Scott, Sandy},
	year = {2008},
	pages = {2959--2977},
}

@inproceedings{utture_student_2019,
  author       = {Akshay Utture and
                  Vedant Somani and
                  Prem Krishnaa and
                  Meghana Nasre},
  title        = {Student Course Allocation with Constraints},
  booktitle    = {Proceedings of SEA{\({^2}\)} 2019: the Special Event on Analysis of Experimental Algorithms},
  series       = {Lecture Notes in Computer Science},
  volume       = {11544},
  pages        = {51--68},
  publisher    = {Springer},
  year         = {2019},
}

@misc{kornbluth_undergraduate_2024,
	title = {Undergraduate {Course} {Allocation} through {Competitive} {Markets}},
	url = {http://arxiv.org/abs/2412.05691},
	doi = {10.48550/arXiv.2412.05691},
	urldate = {2025-01-10},
	publisher = {arXiv},
	author = {Kornbluth, Daniel and Kushnir, Alexey},
	month = dec,
	year = {2024},
	note = {arXiv:2412.05691 [econ]},
}

@book{garey_1979_computers,
  title={Computers and Intractability; A Guide to the Theory of NP-Completeness},
  author={Garey, Michael R and Johnson, David S},
  year={1979},
  publisher={WH Freeman \& Co.}
}

@article{horton_1993_minimum,
  title={Minimum edge dominating sets},
  author={Horton, Joseph Douglas and Kilakos, Kyriakos},
  journal={SIAM Journal on Discrete Mathematics},
  volume={6},
  number={3},
  pages={375--387},
  year={1993},
  publisher={SIAM}
}

@inproceedings{biswas_2023_algorithmic,
  title={An Algorithmic Approach to Address Course Enrollment Challenges},
  author={Biswas, Arpita and Ke, Yiduo and Khuller, Samir and Liu, Quanquan C},
  booktitle={Proceedings of FORC 2023: the 4th Symposium on Foundations of Responsible Computing},
  volume={256},
  pages = {8:1--8:23},
  series = {Leibniz {International} {Proceedings} in {Informatics} ({LIPIcs})},
  publisher = {Schloss Dagstuhl - Leibniz-Zentrum f{\"{u}}r Informatik},
  year={2023},
}

@article{roth_1984_evolution,
  year = {1984},
  title = {The evolution of the labor market for medical interns and residents: a case study in game theory},
  pages = {991-1016},
  number = {6},
  author = {A.E. Roth},
  journal = {Journal of Political Economy},
  volume = {92}
}

@article{ron_1990_np,
  year = {1990},
  title = {{NP}-complete stable matching problems},
  pages = {285-304},
  author = {E. Ronn},
  journal = {Journal of Algorithms},
  volume = {11}
}

@techreport{ng_1988_complexity,
  year = {1988},
  title = {Complexity of the stable marriage and stable roommate problems in three dimensions},
  institution = {Department of Information and Computer Science, University of California, Irvine},
  number = {UCI-ICS 88-28},
  author = {C. Ng and D.S. Hirschberg}
}

@inproceedings{soumalias_machine_2024,
	title = {Machine {Learning}-{Powered} {Course} {Allocation}},
	booktitle = {Proceedings of EC 2024: the 25th {ACM} {Conference} on {Economics} and {Computation}},
	author = {Soumalias, Ermis and Zamanlooy, Behnoosh and Weissteiner, Jakob and Seuken, Sven},
	year = {2024},
    publisher = {ACM},
	note = {arXiv:2210.00954 [cs]},
	pages = {1099--1099},
}

@article{budish_combinatorial_2011,
	title = {The {Combinatorial} {Assignment} {Problem}: {Approximate} {Competitive} {Equilibrium} from {Equal} {Incomes}},
	volume = {119},
	issn = {0022-3808},
	shorttitle = {The {Combinatorial} {Assignment} {Problem}},
	url = {https://www.journals.uchicago.edu/doi/full/10.1086/664613},
	doi = {10.1086/664613},
	number = {6},
	urldate = {2025-05-07},
	journal = {Journal of Political Economy},
	author = {Budish, Eric},
	year = {2011},
	pages = {1061--1103},
}

@article{sonmez_course_2010,
	title = {Course {Bidding} at {Business} {Schools}},
	volume = {51},
	copyright = {© (2010) by the Economics Department of the University of Pennsylvania and the Osaka University Institute of Social and Economic Research Association},
	issn = {1468-2354},
	url = {https://onlinelibrary.wiley.com/doi/abs/10.1111/j.1468-2354.2009.00572.x},
	doi = {10.1111/j.1468-2354.2009.00572.x},
	language = {en},
	number = {1},
	urldate = {2025-05-08},
	journal = {International Economic Review},
	author = {Sönmez, Tayfun and Ünver, M. Utku},
	year = {2010},
	pages = {99--123},
}

@article{budish_multi-unit_2012,
	title = {The {Multi}-unit {Assignment} {Problem}: {Theory} and {Evidence} from {Course} {Allocation} at {Harvard}},
	volume = {102},
	issn = {0002-8282},
	shorttitle = {The {Multi}-unit {Assignment} {Problem}},
	url = {https://www.aeaweb.org/articles?id=10.1257/aer.102.5.2237},
	doi = {10.1257/aer.102.5.2237},
	language = {en},
	number = {5},
	urldate = {2025-05-08},
	journal = {American Economic Review},
	author = {Budish, Eric and Cantillon, Estelle},
	year = {2012},
	pages = {2237--2271},
}

@article{diebold_matching_2017,
	title = {Matching with indifferences: {A} comparison of algorithms in the context of course allocation},
	volume = {260},
	issn = {0377-2217},
	shorttitle = {Matching with indifferences},
	url = {https://www.sciencedirect.com/science/article/pii/S037722171631030X},
	doi = {10.1016/j.ejor.2016.12.011},
	number = {1},
	urldate = {2025-05-08},
	journal = {European Journal of Operational Research},
	author = {Diebold, Franz and Bichler, Martin},
	year = {2017},
	pages = {268--282},
}

@article{atef_yekta_optimization-based_2020,
	title = {Optimization-based {Mechanisms} for the {Course} {Allocation} {Problem}},
	volume = {32},
	issn = {1091-9856},
	url = {https://pubsonline.informs.org/doi/abs/10.1287/ijoc.2018.0849},
	doi = {10.1287/ijoc.2018.0849},
	number = {3},
	urldate = {2025-05-08},
	journal = {INFORMS Journal on Computing},
	author = {Atef Yekta, Hoda and Day, Robert},
	month = jul,
	year = {2020},
	pages = {641--660},
}

@article{roth_1986_allocation,
  title={On the allocation of residents to rural hospitals: a general property of two-sided matching markets},
  author={Roth, Alvin E},
  journal={Econometrica: Journal of the Econometric Society},
  pages={425--427},
  year={1986},
  publisher={JSTOR}
}

@inproceedings{iwama_2002_stable,
  title={Stable marriage with incomplete lists and ties},
  author={Iwama, Kazuo and Manlove, David and Miyazaki, Shuichi and Morita, Yasufumi},
  booktitle={Proceedings of ICALP 1999: the  26th International Colloquium on Automata, Languages and Programming},
  series = {Lecture Notes in Computer Science},
  volume = {1644},
  pages={443--452},
  year={1999},
  organization={Springer}
}

@article{manlove_2002_hard,
  title={Hard variants of stable marriage},
  author={Manlove, David F and Irving, Robert W and Iwama, Kazuo and Miyazaki, Shuichi and Morita, Yasufumi},
  journal={Theoretical Computer Science},
  volume={276},
  number={1-2},
  pages={261--279},
  year={2002},
  publisher={Elsevier}
}

@article{gale_1985_some,
  title={Some remarks on the stable matching problem},
  author={Gale, David and Sotomayor, Marilda},
  journal={Discrete Applied Mathematics},
  volume={11},
  number={3},
  pages={223--232},
  year={1985},
  publisher={Elsevier}
}

@misc{elte_regulations_2025,
  title = {Regulations of ranking},
  author = {{Eötvös Loránd University}},
  howpublished = {\url{https://qter.elte.hu/Statikus.aspx}},
  note = {Accessed: 2025-05-12},
}

@misc{glasgow_catalogue_2025,
  title = {{University of Glasgow -- Course Catalogue}},
  author = {{University of Glasgow}},
  howpublished = {\url{https://www.gla.ac.uk/coursecatalogue/}},
  note = {Accessed: 2025-05-14},
}

@misc{glasgowlaw_communication_2025,
  author = {Graves, Zoe and Wierzbicki, Matt},
  year = {2025},
  howpublished = {Personal communication},
}

@techreport{biro_college_2009,
	title = {The {College} {Admissions} problem with lower and common quotas},
	url = {https://www.dcs.gla.ac.uk/~davidm/pubs/9150.pdf},
	doi = {10.1016/j.tcs.2010.05.005},
	urldate = {2025-05-15},
    type = {Technical Report},
    number = {TR-2009-303},
	institution = {School of Computing Science, University of Glasgow},
	author = {Biró, Péter and Fleiner, Tamás and Irving, Robert W. and Manlove, David F.},
	year = {2009},
}

@article{cechlarova_pareto_2017,
    title = {Pareto optimal matchings with lower quotas},
    journal = {Mathematical Social Sciences},
    volume = {88},
    pages = {3-10},
    year = {2017},
    issn = {0165-4896},
    doi = {https://doi.org/10.1016/j.mathsocsci.2017.03.007},
    url = {https://www.sciencedirect.com/science/article/pii/S0165489617300641},
    author = {Katarína Cechlárová and Tamás Fleiner},
}

@article{nguyen_2018_near,
  title={Near-feasible stable matchings with couples},
  author={Nguyen, Thanh and Vohra, Rakesh},
  journal={American Economic Review},
  volume={108},
  number={11},
  pages={3154--3169},
  year={2018},
  publisher={American Economic Association 2014 Broadway, Suite 305, Nashville, TN 37203}
}

@techreport{delacretaz_2019_stability,
  title={Stability in matching markets with sizes},
  author={Delacr{\'e}taz, David},
  year={2019},
  institution={Working paper},
  note = {\url{https://daviddelacretaz.net/wp-content/uploads/2019/01/DDelacretaz-SMMS-Jan19.pdf}}
}

@article{delacretaz_2023_matching,
  Author = {Delacrétaz, David and Kominers, Scott Duke and Teytelboym, Alexander},
  Title = {Matching Mechanisms for Refugee Resettlement},
  Journal = {American Economic Review},
  Volume = {113},
  Number = {10},
  Year = {2023},
  Pages = {2689–2717},
  DOI = {10.1257/aer.20210096},
  URL = {https://www.aeaweb.org/articles?id=10.1257/aer.20210096}
}

@InProceedings{balasundaram_stability_2025,
  author =	{Balasundaram, Haricharan and Krishnashree, J. B. and Limaye, Girija and Nasre, Meghana},
  title =	{{Stability Notions for Hospital Residents with Sizes}},
  booktitle =	{Proceedings of FSTTCS 2025: the 45th IARCS Annual Conference on Foundations of Software Technology and Theoretical Computer Science},
  pages =	{11:1--11:18},
  series =	{Leibniz International Proceedings in Informatics (LIPIcs)},
  year =	{2025},
  volume =	{360},
  publisher =	{Schloss Dagstuhl -- Leibniz-Zentrum f{\"u}r Informatik},
}

\end{document}